\documentclass[article, DIV=11]{scrartcl}

\usepackage[utf8]{inputenc}
\usepackage[T1]{fontenc}
\usepackage{lmodern}
\usepackage{textcomp}
\usepackage{bbm}

\usepackage{amsmath, amssymb, amsthm}

\usepackage[table]{xcolor}
\usepackage{booktabs}    
\usepackage{colortbl}    
\usepackage{multirow}    
\usepackage{array}       
\usepackage{longtable}   
\usepackage[round]{natbib}

\usepackage{microtype}
\KOMAoptions{parskip=half}
\usepackage{graphicx}
\usepackage{float}       
\usepackage{caption}
\usepackage{subcaption}  

\usepackage{hyperref}
\hypersetup{
  pdftitle={Causal survival analysis},
  pdfauthor={Charlotte Voinot; Clement Berenfeld; Imke Mayer; Bernard Sebastien; Julie Josse},
  pdfkeywords={Restricted Mean Survival Time, Randomized Control Trial, Observational Study, Censoring},
  colorlinks=true,
  linkcolor=blue,
  citecolor=blue,
  urlcolor=blue
}

\renewcommand{\epsilon}{\varepsilon}
\renewcommand{\leq}{\leqslant}
\renewcommand{\geq}{\geqslant}

\newtheorem{theorem}{Theorem}[section]
\newtheorem{corollary}{Corollary}[section]
\newtheorem{proposition}{Proposition}[section]
\newtheorem{definition}{Definition}[section]

\newtheorem{assumption}{Assumption}[section] 

\usepackage{etoolbox}
\usepackage{hyperref}

\makeatletter
\newcommand{\assumptionname}[1]{%
  \gdef\@currentlabelname{#1}%
}
\makeatother

\title{Treatment Effect Estimation in Causal Survival Analysis: Practical Recommendations}
\author{Charlotte Voinot \and Clément Berenfeld \and Imke Mayer \and Bernard Sebastien \and Julie Josse}
\date{}

\begin{document}
\maketitle
\begin{abstract}
\textcolor{black}{The restricted mean survival time (RMST) difference offers an interpretable causal contrast to estimate the treatment effect for time-to-event outcomes, yet a wide range of available estimators leaves limited guidance for practice. We provide a unified review of RMST estimators for randomized trials and observational studies, establish identification and asymptotic properties, and supply new derivations where needed. Our extensive simulation study compares simple nonparametric methods (such as unweighted Kaplan–Meier estimators) alongside parametric and nonparametric implementations of the G-formula, weighting approaches, Buckley–James transformations, and augmented estimators under diverse censoring mechanisms and model specifications.
Across scenarios, classical Kaplan–Meier estimators (weighted when required by the censoring process) and G-formula methods perform well in randomized settings, while in observational data G-formula estimators remain competitive; however, augmented estimators such as AIPTW–AIPCW generally offer robustness to model misspecification and a favorable bias–variance trade-off. Parametric estimators perform best under correct specification, whereas nonparametric methods avoid functional assumptions but require large sample sizes to achieve reliable performance. We offer practical recommendations for estimator choice and provide open-source R code to support reproducibility and application.}

\end{abstract}

\section{Introduction}\label{sec-intro}

Causal survival analysis is a rapidly evolving field that integrates causal inference \citep{rubin_1974, Hernan2020} with survival analysis \citep{kalbfleisch2002} to evaluate the impact of treatments on time-to-event outcomes, while accounting for censoring, i.e. situations where only partial information about an event's occurrence is available.
In this litterature, the Restricted Mean Survival Time (RMST) difference has emerged as a promising causal measure  that offers an intuitive interpretation of the average survival time over a specified time horizon. 

\textcolor{black}{Despite its recent emergence, the field has expanded quickly, though often in a fragmented manner. As a result, researchers are faced with a broad spectrum of estimators, from weighting and regression-based approaches to more complex hybrid strategies, and must navigate numerous methodological decisions. These include, for example, the choice of how to estimate nuisance components, whether through parametric models, or nonparametric techniques, frequently without clear guidance. The situation is further complicated by the limited availability of accessible software implementations. Addressing these challenges is especially crucial in real-world contexts where randomized trials are not feasible and carefully designed observational studies, grounded in explicit causal assumptions, represent the primary source of evidence on treatment effects \citep{Makary2025}. These considerations underscore the need for a structured review that consolidates existing approaches, clarifies their underlying assumptions, and evaluates estimator performance to support informed methodological decisions, particularly in finite-sample settings. Such a survey could also help to identify key methodological gaps that warrant further research.}

\textcolor{black}{Related efforts have already been made to clarify methodological options in survival analysis. For instance, \citet{denz2023} focus on confounders-adjusted survival curves under proportional hazards and assess the impact of model misspecification. Yet, most existing reviews remain confined to the Cox model framework and thus overlook the broader range of estimators and causal estimands now available.}

\textcolor{black}{
Despite its widespread use, we chose not to cover the estimation of the hazard ratio (HR), typically derived from a Cox model, because of several well-documented limitations. It relies on the proportional hazards assumption, which may not hold in practice. Moreover, the HR does not represent a causal contrast \citep{Martinussen2020}, even in randomized trials, due to inherent selection bias that can distort group comparisons unless the treatment has no effect \citep{Martinussen2013, Hernan2010_HR}. The HR is also non-collapsible, meaning marginal effects cannot be inferred from conditional ones \citep{huitfeldt2019collapsibility, greenland1999}. Finally, its clinical interpretation could be difficult, as it reflects an instantaneous risk rather than a cumulative or absolute survival benefit \citep{averbuch2025}.}

In this paper, after introducing the notations and the definition of the restricted mean survival time (Section~\ref{sec-notations}), we present the statistical framework required to define and estimate RMST in both randomized controlled trials (Section~\ref{sec-theoryRCT}) and observational settings (Section~\ref{sec-theoryOBS}). We characterize their key statistical properties, such as consistency and asymptotic normality, and provide a unified, self-contained set of proofs. Where existing results are incomplete or unavailable, we supply original derivations; where they exist, we rederive them within a common formalism that enhances clarity and facilitates direct comparison between estimators. In Section~\ref{sec-simulation}, we complement this theoretical analysis with a simulation study assessing estimator performance under various conditions, including independent and conditionally independent censoring and both correctly specified and misspecified models. Finally, in Section~\ref{sec-conclusion}, we offer practical recommendations for estimator selection, grounded in considerations such as finite-sample behavior and computational complexity.

\subsection{Notations and definition of the Restricted Mean Survival Time}\label{sec-notations}

We set the analysis in the potential outcome framework, where a patient,
described by a vector of covariates \(X \in \mathbb{R}^p\), either
receives a treatment (\(A=1\)) or is in the control group (\(A=0\)). 
The patient’s survival time would then be \(T^{(0)} \in \mathbb{R}^+\) under control and \(T^{(1)} \in \mathbb{R}^+\) under treatment. In practice, both potential outcomes cannot be observed simultaneously for the same individual, since each patient receives only one treatment condition. Instead, we observe the realized survival time \(T\), defined as follows:
\begin{assumption}
\phantomsection
\assumptionname{SUTVA}
\label{ass:sutva}
\textbf{(Stable Unit Treatment Value Assumption — SUTVA)} $T = A T^ {(1)} + (1 - A) T^ {(0)}$.
\end{assumption}

Because of potential censoring, the outcome \(T\) is not always fully observed. The most common form is right-censoring, which occurs when the event has not occurred by the end of follow-up \citep{Turkson2021}. In what follows, we focus on this case: for each subject, we observe an \(n\)-sample \((X, A, \tilde T, \Delta)\), where \(\tilde T = \min(T, C)\), \(C \in \mathbb{R}_+\), and \(\Delta = \mathbb{I}\{T \le C\}\) (equal to \(1\) if the event time is observed and \(0\) otherwise). These observed variables are derived from the  unobservable tuple \((X, A, T^{(0)}, T^{(1)}, C)\) illustrated in Table~\ref{tbl-example_data}.

\begin{table}[h!]
\centering
\caption{Example of a survival dataset. Each observation corresponds to a patient. In practice, only \(X\), \(A\), \(\Delta\), and \(\widetilde{T}\) are observed. The column \textit{Obs. Outcome} represents the observed survival outcome.}
\label{tbl-example_data}
\renewcommand{\arraystretch}{1.2}
\setlength{\tabcolsep}{6pt}

\hspace*{-0.9cm}
\begin{tabular}{|
|>{\columncolor{gray!15}}c|>{\columncolor{gray!15}}c|
>{\columncolor{gray!15}}c|
>{\columncolor{gray!15}}c|
c|c|c|c|
>{\columncolor{gray!15}}c|
>{\columncolor{gray!15}}c||}
\hline
\multicolumn{3}{||>{\columncolor{gray!15}}c|}{\small \textbf{Covariates}} 
& \textbf{\small Treatment} 
& \multicolumn{2}{c|}{\small \textbf{Potential outcome}} 
& \textbf{\small Outcome} 
& \textbf{\small Censoring} 
& \textbf{\small Status} 
& \textbf{\small Obs. Outcome} \\
 \(X_1\) & \(X_2\) & \(X_3\) 
& \(A\) 
& \(\quad T^ {(0)} \; \; \) & \(T^ {(1)}\) 
& \(T\) 
& \(C\) 
& \(\Delta\) 
& \(\widetilde{T}\) \\
\hline
1 & 1.5 & 4 & 1 & ? & 200 & 200 & ? & 1 & 200 \\
5 & 1 & 2 & 0 & 100 & ? & 100 & ? & 1 & 100 \\
9 & 0.5 & 3 & 1 & ? & ? & ? & 200 & 0 & 200 \\
\hline
\end{tabular}
\end{table}

Our aim is to estimate the Average Treatment Effect (ATE) defined as the difference between the Restricted Mean Survival Time \citep{royston2013}.

\begin{definition}[] (Causal effect: Difference between Restricted Mean Survival Time)\protect\hypertarget{def-ATE}{}\label{def-ATE} We define 
$\theta_{\mathrm{RMST}}=\mathbb{E}\!\left[T^{(1)} \wedge \tau - T^{(0)} \wedge \tau\right],$
where \(a \wedge b := \min(a,b)\) for \(a,b \in \mathbb{R}\), and \(\tau>0\) is a pre-specified time horizon over which survival is evaluated.
\end{definition}

We define the survival functions \(S^{(a)}(t):=\mathbb{P}(T^ {(a)} > t)\)
for \(a \in \{0,1\}\), i.e., the probability that an individual will survive beyond a given time \(t\) if treated or if not treated. Likewise,
we let \(S(t) := \mathbb{P}(T >t)\), and
\(S_C(t) := \mathbb{P}(C > t)\). We also let
\(G(t) := \mathbb{P}(C \geqslant t)\) be the left-limit of the survival
function \(S_C\). Because \(T^ {(a)} \wedge \tau\) are non-negative random
variables, one can easily express the restricted mean survival time
using the survival functions:
\begin{equation}\phantomsection\label{eq-rmstsurv}{
\mathbb{E}(T^ {(a)} \wedge \tau) = \int_{0}^{\infty} \mathbb{P}(T^ {(a)}\wedge \tau > t)dt = \int_{0}^{\tau}S^{(a)}(t)dt.
}\end{equation}
Consequently, \(\theta_{\mathrm{RMST}}\) can be interpreted as the mean
difference between the survival function of treated and control until a
fixed time horizon \(\tau\) (See Figure~\ref{fig-RMST}). A RMST difference
\(\theta_{\mathrm{RMST}} = 10\) days with \(\tau=200\) means that, on
average, the treatment increases the survival time by 10 days at 200
days.

Although the present work focuses on the estimation of the difference of RMST, we would like to stress that the causal effect can be assessed through other measures, such as for instance the difference between the survival functions $\theta_{\mathrm{SP}} := S^{(1)}(\tau) - S^{(0)}(\tau)$ for some time \(\tau\) \citep{Ozenne20_AIPTW_AIPCW}. 

\begin{figure}[H]
\centering{
\hspace*{-0.05\textwidth}
\includegraphics[width=0.5\textwidth]{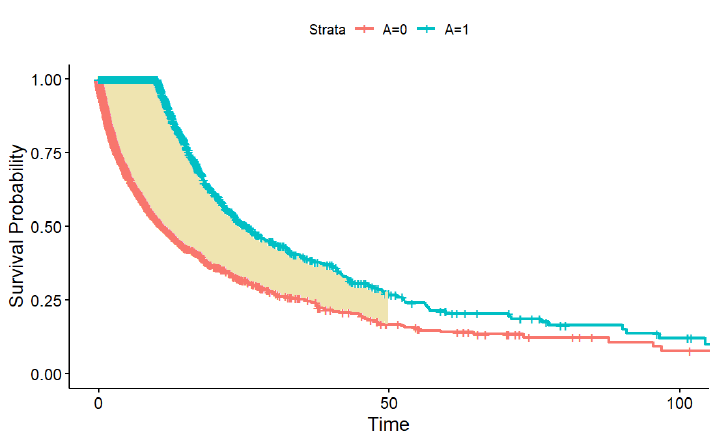}
}
\caption{\label{fig-RMST}Estimated survival curves on synthetic data.  The \(\theta_{\mathrm{RMST}}\) at \(\tau=50\)
corresponds to the yellow area between them.
Curves are estimated with the Kaplan-Meier estimator, see
Section~\ref{sec-theoryRCT_indc}.}
\end{figure}%

We provide in Table~\ref{tbl-reminder_notations} a summary of the
notations used throughout the paper.
\begin{longtable}[]{@{}
  >{\raggedright\arraybackslash}p{0.25\linewidth}
  >{\raggedright\arraybackslash}p{0.75\linewidth}@{}}
\caption{Summary of the notations.}
\label{tbl-reminder_notations}\\
\toprule\noalign{}
\begin{minipage}[b]{\linewidth}\raggedright
Symbol
\end{minipage} & \begin{minipage}[b]{\linewidth}\raggedright
Description
\end{minipage} \\
\midrule\noalign{}
\endfirsthead

\toprule\noalign{}
\begin{minipage}[b]{\linewidth}\raggedright
Symbol
\end{minipage} & \begin{minipage}[b]{\linewidth}\raggedright
Description
\end{minipage} \\
\midrule\noalign{}
\endhead

\bottomrule\noalign{}
\endlastfoot

\(X\) & Covariates \\
\(A\) & Treatment indicator \((A=1\) for treatment, \(A=0\) for
control\()\) \\
\(T\) & Survival time \\
\(T^ {(a)}, a \in \{0,1\}\) & Potential survival time respectively with and
without treatment \\
\(S^{(a)},a \in \{0,1\}\) & Survival function
\(S^{(a)}(t) =\mathbb{P}(T^ {(a)} > t)\) of the potential survival times \\
\(\lambda^{(a)},a \in \{0,1\}\) & Hazard function
\(\lambda^{(a)}(t) =\lim_{h \to 0^+}\mathbb{P}(T^ {(a)} \in [t,t+h) |T^ {(a)}\geqslant t)/h\)
of the potential survival times \\
\(\Lambda^{(a)},a \in \{0,1\}\) & Cumulative hazard function \(\Lambda^{(a)}(t) := \int_0^t \lambda^{(a)}(s)\mathop{}\!\mathrm{d}s\) of the
potential survival times \\
\(C\) & Censoring time \\
\(S_C\) & Survival function \(S_C(t) =\mathbb{P}(C > t)\) of the
censoring time \\
\(G\) & Left-limit of the survival function
\(G(t) =\mathbb{P}(C \geqslant t)\) of the censoring time \\
\(\widetilde{T}\) & Observed time (\(T \wedge C=\min(T,C)\)) \\
\(\Delta\) & Event status \(\mathbb{I}\{T \leqslant C \}\) \\
\(\Delta^{\tau}\) & Event status of the restricted time
\(\Delta^{\tau} = \max\{\Delta, \mathbb{I}\{\widetilde{T} \geqslant\tau\}\}\) \\
\(\{t_{1},t_{2},\dots,t_{K}\}\) & Discrete times \\
\(e(x)\) & Propensity score \(\mathbb{E} [A| X = x]\) \\
\(\mu(x,a), a \in \{0,1\}\) &
\(\mathbb{E}[T \wedge \tau \mid X=x,A=a ]\) \\
\(S(t|x,a), a \in \{0,1\}\) & Conditional survival function,
\(\mathbb{P}(T > t | X=x, A =a)\). \\
\(\lambda^{(a)}(t|x), a \in \{0,1\}\) & Conditional hazard functions of
the potential survival times \\
\(G(t|x,a), a \in \{0,1\}\) & left-limit of the conditional survival
function of the censoring \(\mathbb{P}(C\geqslant t|X=x,A=a)\) \\
\(Q_{S}(t|x,a), a \in \{0,1\}\) &
\(\mathbb{E}[T \wedge \tau \mid X=x,A=a, T \wedge \tau>t]\) \\
\end{longtable}

\section{Causal survival analysis in Randomized Controlled Trials}\label{sec-theoryRCT}

Randomized Controlled Trials (RCTs) are the gold standard for
establishing the effect of a treatment on an outcome, because treatment
allocation is controlled through randomization, which ensures
(asymptotically) the balance of covariates between treated and controls,
and thus avoids problems of confounding between treatment groups. In a classical RCT, treatment is randomly assigned \citep{rubin_1974}, and everyone could, in principle, be assigned to either arm.

\begin{assumption}
\phantomsection
\assumptionname{Random Treatment Assignment}
\label{ass:rta}
\textbf{(Random Treatment Assignment)}
$A \perp\mkern-9.5mu\perp T^ {(0)}, T^ {(1)}, X$.
\end{assumption}

\begin{assumption}
\phantomsection
\assumptionname{Trial Positivity}
\label{ass:trialpositivity}
\textbf{(Trial Positivity)}
$0 < \mathbb{P}(A=1) < 1$.
\end{assumption}

\noindent Under Assumptions~\ref{ass:sutva} (\nameref{ass:sutva}), \ref{ass:rta} (\nameref{ass:rta}) and \ref{ass:trialpositivity} (\nameref{ass:trialpositivity}), classical causal
identifiability equations can be written to express
\(\theta_{\mathrm{RMST}}\) without potential outcomes.
\begin{equation}\phantomsection\label{eq-RMSTkm}{
\begin{aligned}
    \theta_{\mathrm{RMST}} &=  \mathbb{E}[T^ {(1)} \wedge \tau - T^ {(0)} \wedge \tau]\\
    &= \mathbb{E}[T^ {(1)} \wedge \tau | A = 1] - \mathbb{E}[ T^ {(0)} \wedge \tau| A= 0]  && \text{(Random treatment assignment)} \\
       &= \mathbb{E}[T \wedge \tau | A = 1] - \mathbb{E}[ T \wedge \tau| A= 0].  && \text{(SUTVA)}
\end{aligned}
}\end{equation}

\noindent 
However, Equation~\ref{eq-RMSTkm} still depends on \(T\), which is only partially observed due to censoring. To ensure identifiability of treatment effects, standard assumptions on the censoring mechanism are required, namely independent and conditionally independent censoring. 

When censoring is independent, standard survival estimators, such as the Kaplan–Meier estimator, directly accommodate censoring, allowing RMST to be obtained from the estimated survival curve via Equation~\ref{eq-rmstsurv} (Section~\ref{sec-theoryRCT_indc}). Under the more general condition of conditionally independent censoring, censoring-unbiased transformations of the event time can be used to recover \(\mathbb{E}[T \wedge \tau \mid A = a]\)  (Section~\ref{sec-condcens}). These transformations enable two estimation strategies:
(i) direct estimation, by computing the transformed outcome and averaging it and (ii) adjusted survival-curve estimation, in which the transformation is incorporated within a survival estimator, and RMST is then derived using Equation~\ref{eq-rmstsurv}. These two families of methods form the basis of the estimators reviewed in the next sections.

\subsection{Independent censoring: the Kaplan-Meier estimator}\label{sec-theoryRCT_indc}

\begin{assumption}
\phantomsection
\assumptionname{Independent Censoring}
\label{ass:independentcensoring}
\textbf{(Independent Censoring)}
$C \perp\mkern-9.5mu\perp T^ {(0)}, T^ {(1)}, X, A.$
\end{assumption}
\noindent Under Assumption~\ref{ass:independentcensoring}, subjects censored at time
\(t\) are representative of all subjects who remain at risk at time
\(t\). 
\textcolor{black}{Figure~\ref{fig-RCT_ind_causalgraph} shows the causal graph where the treatment assignment is random, the time $T$ depends on \((A,X)\), and the censoring time \(C\) is independent of all variables.
}

\begin{figure}[h]
  \centering
  \includegraphics[width=0.25\textwidth]{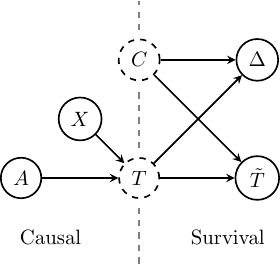}
  \caption{\label{fig-RCT_ind_causalgraph}Causal graph in RCT survival data with independent censoring.}
\end{figure}

We assume that there is no almost-sure upper bound on the censoring
time before \(\tau\).

\begin{assumption}
\phantomsection
\assumptionname{Positivity of the Censoring Process}
\label{ass:poscen}
\textbf{(Positivity of the Censoring Process)}
There exists \(\varepsilon > 0\) such that $G(\tau) \geqslant \varepsilon$.
\end{assumption}
\noindent If indeed it was the case that \(\mathbb{P}(C < t) = 1\) for some
\(t < \tau\), then we would not be able to infer anything on the
survival function on the interval \([t,\tau]\) as all observation times
\(\widetilde T_i\) would be in \([0,t]\) almost surely. In practice,
adjusting the threshold time \(\tau\) can help satisfy this positivity
assumption. 

\textcolor{black}{Under Assumptions~\ref{ass:sutva} (\nameref{ass:sutva}), \ref{ass:rta} (\nameref{ass:rta}), \ref{ass:trialpositivity} (\nameref{ass:trialpositivity}), \ref{ass:independentcensoring} (\nameref{ass:independentcensoring}), and \ref{ass:poscen} (\nameref{ass:poscen}), the potential-outcome survival $S^{(a)}(t)$ is identified as a functional of the joint law of ($\tilde{T}, \Delta, A$)}. This enables several estimation strategies, the most well-known being the Kaplan-Meier, the Maximum Likelihood
Estimator (MLE) of the survival functions, see for instance \cite{Kaplan1958_1}. 

\begin{definition}[]\protect\hypertarget{def-km}{}\label{def-km}

(Kaplan-Meier estimator, \citep{Kaplan1958_1})

$$
    \widehat{S}_{\mathrm{KM}}(t|A=a) := \prod_{t_k \leqslant t}\left(1-\frac{D_k(a)}{N_k(a)}\right).
$$

with $D_k(a) := \sum_{i=1}^n \mathbb{I}(\widetilde T_i = t_k, \Delta_i = 1, A=a) \quad\text{and}\quad N_k(a) := \sum_{i=1}^n \mathbb{I}(\widetilde T_i \geqslant t_k, A=a),$ which correspond respectively to the number of deaths \(D_k(a)\) and of individuals at risk \(N_k(a)\) at time \(t_k\) in the treated group ($a=1$) or in the control group ($a=0$).
\end{definition}

Furthermore, since \(D_k(a)\) and \(N_k(a)\) are sums of i.i.d. random variables, the Kaplan--Meier estimator inherits convenient statistical properties. \textcolor{black}{In Appendix~\ref{sec-proof21}, we show that it exhibits exponentially fast bias decay (Proposition~\ref{prp-km}) and asymptotic normality with explicitly characterized variances (Proposition~\ref{prp-varkm}).} Finally, the estimator in Definition~\ref{def-km} extends naturally to the continuous setting by replacing the fixed grid \(\{t_k\}\) with observed event times \((\widetilde T_i = t_k, \Delta_i = 1)\). We refer to \cite{Breslow1974} for a study of this version and to \cite[Sec.~3.2]{Aalen2008} for a broader point-process perspective.

The associated RMST estimator is then  defined as
\begin{equation}\phantomsection\label{eq-thetaunadjKM}{
\widehat{\theta}_{\mathrm{KM}} = \int_{0}^{\tau}\widehat{S}_{\mathrm{KM}}(t|A=1)-\widehat{S}_{\mathrm{KM}}(t|A=0)dt.
}\end{equation} 

\textcolor{black}{This estimator inherits KM properties. It is consistent for the target RMST and its bias decays exponentially fast. In RCTs, the independence of arms delivers asymptotic normality with variance equal to the sum of the arm-specific RMST variances.}

\subsection{Conditionally independent censoring}\label{sec-condcens}

Conditionally independent censoring relaxes the assumption of independent censoring by allowing censoring to depend on treatment or covariates — for instance, when treated patients withdraw from a study due to side effects.

\begin{assumption}
\phantomsection
\assumptionname{Conditionally Independent Censoring}
\label{ass:condindepcensoring}
\textbf{(Conditionally Independent Censoring)}
$C \perp\mkern-9.5mu\perp T^ {(0)}, T^ {(1)} \mid X, A$
\end{assumption}

Under Assumption~\ref{ass:condindepcensoring}, subjects within a same
stratum defined by \(X=x\) and \(A=a\) have equal probability of
censoring at time \(t\), for all \(t\). In case of conditionally
independent censoring, we also need to assume that all subjects have a
positive probability to remain uncensored at their time-to-event.

\begin{assumption}
\phantomsection
\assumptionname{Conditional Positivity / Overlap for Censoring}
\label{ass:positivitycensoring}
\textbf{(Conditional Positivity / Overlap for Censoring)}
There exists \(\varepsilon> 0\) such that for all \(t \in [0,\tau]\), it almost surely holds $G(t|A,X) \geqslant\varepsilon.$
\end{assumption}

\textcolor{black}{Figure~\ref{fig-RCT_dep_causalgraph} represents the causal graph when treatment assignment is randomized and when there is a conditionally independent censoring mechanism.}

\begin{figure}[h]
  \centering
  \includegraphics[width=0.25\textwidth]{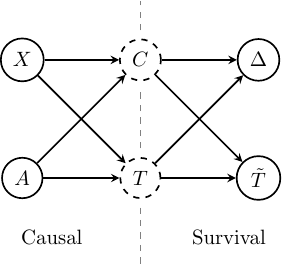}
  \caption{\label{fig-RCT_dep_causalgraph}Causal graph in RCT survival
data with dependent censoring.}
\end{figure}

Under dependent censoring, the Kaplan-Meier estimator as defined in
Definition~\ref{def-km} is not suited to estimate survival. Alternatives
include plug-in G-formula estimators (Section~\ref{sec-Gformula}) and
unbiased transformations (Section~\ref{sec-unbiasedtransfo}).

\subsubsection{The G-formula}\label{sec-Gformula}

Building on Equation~\ref{eq-RMSTkm}, one can express the RMST conditioning on the covariates: \begin{align*}
 \theta_{\mathrm{RMST}} = \mathbb{E}\left[\mathbb{E}[T \wedge \tau |X, A = 1]\right] - \mathbb{E}\left[ T \wedge \tau|X, A= 0]\right] = \mathbb{E}[\mu(X,1)-\mu(X,0)].
\end{align*}
with $\mu(x,a) := \mathbb{E}[T \wedge \tau |X = x, A = a].$

An estimator of the RMST difference via the \emph{G-formula} arises by plugging in any estimator \(\widehat\mu\) of \(\mu\):
\begin{equation}\phantomsection\label{eq-gformula}{ \widehat{\theta}_{\mathrm{G-formula}} = \frac{1}{n} \sum_{i=1}^n \widehat\mu\left(X_i, 1\right)-\widehat\mu\left(X_i, 0\right). }\end{equation}

\textcolor{black}{The G-formula accommodates right censoring under conditionally independent censoring, without specifying a separate censoring model: the target (e.g., \(S^{(a)}(t)\) or \(\mathrm{RMST}^{(a)}\)) is expressed as a functional of the conditional outcome regression \(\mu(\cdot)\).
The same G-formula can be applied to observational settings (in Section~\ref{sec-obscondcens}) to address non-random treatment assignment. Different strategies can be used to estimate \(\mu\), ranging from semi-parametric models such as the Cox model to fully non-parametric approaches like survival forests.
}
\textcolor{black}{For any chosen model, implementation follows one of two paradigms: \emph{T-learner} approaches (i.e., two separate learners), where \(\mu(\cdot,1)\) and \(\mu(\cdot,0)\) are estimated separately; and \emph{S-learner} approaches (i.e., a single learner), where \(\widehat{\mu}\) is obtained by fitting one model using covariates \((X,A)\).}\\

\textbf{Cox Model}. The Cox proportional hazards
model \citep{Cox1972} relies on a semi-parametric model of the
conditional hazard functions and specifies the function under treatment $a$,  \(\lambda^{(a)}(t|X)\), as 
\[
\lambda^{(a)}(t|X) = \lambda_0^{(a)}(t) \exp(X^\top\beta^{(a)}),
\] 
where \(\lambda^{(a)}_0\) is a baseline hazard function under treatment $a$ and
\(\beta^{(a)}\) has the same dimension as the vector of covariate \(X\). The model is semi-parametric in the sense that the baseline hazard is left unspecified, while the effect of covariates enters through a multiplicative proportional hazards structure. \textcolor{black}{Alternative semi-parametric models include Accelerated Failure Time (AFT) \citep{saikia2017} which directly models the survival time and does not rely on the proportional hazards (PH) assumption.}\\

\textbf{Weibull Model}. A popular parametric model for survival is the
Weibull Model, which amounts to assume that $\lambda^{(a)}(t|X) =  t^{\gamma (a)} \exp(X^\top\beta)$ where \(\gamma(a) >0\). We refer to \cite{zhang2016} for a study of this model.\\

\textbf{Survival Forests}. On the nonparametric side, survival forests \citep{ishwaran2008} aggregate survival trees that use log-rank splitting and terminal-node Nelson–Aalen estimates to learn \(S(t\mid X)\) under right censoring, capturing nonlinearities and high-order interactions without parametric form.

\subsubsection{Censoring-unbiased transformations (CUT) }\label{sec-unbiasedtransfo}
Censoring-unbiased transformations involve applying a transformation to
\(T\). Specifically, we compute a new time \(T^*\) of the form
\begin{equation}\phantomsection\label{eq-defcut}{
T^* := T^*(\widetilde T,X,A,\Delta) = \begin{cases}
\phi_0(\widetilde T \wedge \tau,X,A) \quad &\text{if} \quad \Delta^\tau = 0, \\
\phi_1(\widetilde T \wedge \tau,X,A) \quad &\text{if} \quad \Delta^\tau = 1.
\end{cases}
}\end{equation} 
for two wisely chosen transformations \(\phi_0\) and
\(\phi_1\), and where $\Delta^{\tau}:=\mathbb{I}\{T \wedge \tau \leqslant C\} = \Delta+(1-\Delta)\mathbb{I}(\widetilde T \geqslant\tau)$ indicates whether an individual is either uncensored or censored after time \(\tau\). Since we only need the expectation of survival time thresholded at \(\tau\), any censoring beyond \(\tau\) can be treated as uncensored
 (\(\Delta^\tau = 1\)).
A censoring-unbiased transformation \(T^*\) is defined such that, for any \(a \in \{0,1\}\),
\begin{equation}\phantomsection\label{eq-cut}{
\mathbb{E}[T^*|A=a,X] = \mathbb{E}[T^ {(a)} \wedge \tau |X] \quad\text{almost surely.}
}\end{equation} A notable advantage of this approach is that it enables
the use of the full transformed dataset \((X_i,A_i,T^*_i)\) as if no
censoring occurred. 
Moreover, by the reweighting identity
\begin{equation}\phantomsection\label{eq-cutcond}{
\mathbb{E}[\mathbb{E}[T^*|A=a,X]] = \mathbb{E}\left[\frac{\mathbb{I}\{A=a\}}{\mathbb{P}(A=a)} T^*\right],
}\end{equation} this equality links the arm-specific mean of \(T^*\) to its full-sample inverse-probability–weighted representation. The RMST difference can be estimated by \begin{equation}\phantomsection\label{eq-cutest}{
\widehat\theta = \frac1n\sum_{i=1}^n \left(\frac{A_i}{\pi}-\frac{1-A_i}{1-\pi} \right) T^*_i,
}\end{equation} where \(\pi = \mathbb{P}(A=1) \in (0,1)\) by Assumption
\ref{ass:trialpositivity} and where
\(T^*_i = T^*(\widetilde T_i,X_i,A_i,\Delta_i)\). 

The following proposition establishes the theoretical guarantees for the class of estimators constructed from bounded censoring-unbiased transformations.

\begin{proposition}[Consistency of CUT as RMST estimator with known $\pi$]\protect\hypertarget{prp-cutest}{}\label{prp-cutest}
Under Assumptions \ref{ass:rta} (\nameref{ass:rta}) and \ref{ass:trialpositivity} (\nameref{ass:trialpositivity}), the estimator
\(\widehat\theta\) derived as in Equation~\ref{eq-cutest} from square
integrable censoring-unbiased transformations satisfying
Equation~\ref{eq-cut} is an unbiased, strongly consistent, and
asymptotically normal estimator of the RMST difference.

\end{proposition}

Square integrability will be ensured any time the transformation is
bounded, which will always be the case of the ones considered in this
work. 

It is natural in a RCT setting to assume that probability of being
treated \(\pi\) is known. If not, it is usual to replace \(\pi\) by its
empirical counterpart \(\widehat\pi = n_1/n\) where
\(n_a = \sum_{i} \mathbbm{1}\{A=a\}\). The resulting estimator takes the
form \begin{equation}\phantomsection\label{eq-cutestnopi}{
\widehat\theta = \frac1{n_1}\sum_{A_i = 1}  T^*_i - \frac1{n_0}\sum_{A_i = 0}  T^*_i.
}\end{equation} While this estimator introduces a slight bias due to the estimation of \(\pi\) (see for instance \cite{colnet2022}, Lemma 2), it remains strongly consistent and asymptotically normal under the same assumptions:

\begin{proposition}[Consistency of CUT as RMST estimator with estimated $\hat{\pi}$]\protect\hypertarget{prp-cutestnopi}{}\label{prp-cutestnopi}

Under Assumptions \ref{ass:rta} (\nameref{ass:rta}) and \ref{ass:trialpositivity} (\nameref{ass:trialpositivity}), the estimator
\(\widehat\theta\) derived as in Equation~\ref{eq-cutestnopi} from a
square integrable censoring-unbiased transformations satisfying
Equation~\ref{eq-cut} is a strongly consistent, and asymptotically
normal estimator of the RMST difference.

\end{proposition}

For the later, we consider only estimator with estimated probability of being treated in Equation~\ref{eq-cutestnopi}.

The two most popular transformations are
Inverse-Probability-of-Censoring Weighting \citep{koul1981} and Buckley-James \citep{buckley1979}, both illustrated
in Figure~\ref{fig-trans} and detailed below. 
\begin{figure}[h]
\centering{
\includegraphics[width=0.9\textwidth]{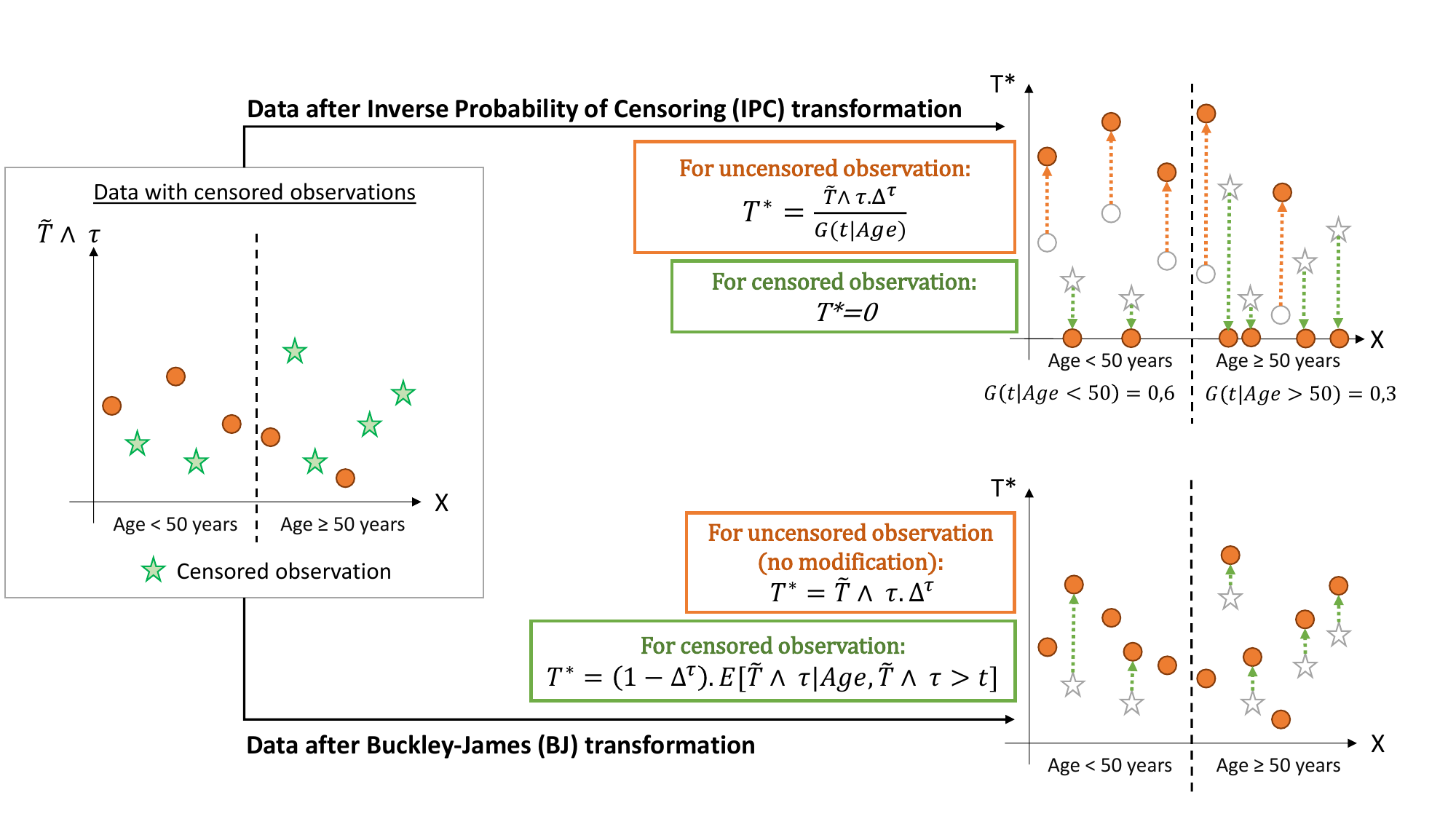}
}
\caption{\label{fig-trans}Illustration of Inverse-Probability-of-Censoring and Buckley-James transformations.}
\end{figure}%
In the former, only non-censored observations are considered and they are weighted while in
the latter, censored observations are imputed with an estimated survival
time.\\

\textbf{The Inverse-Probability-of-Censoring Weighted transformation}

The Inverse-Probability-of-Censoring Weighted (IPCW) transformation,
introduced by \cite{koul1981} in the context of
censored linear regression, involves discarding censored observations
and applying weights to uncensored data (see Figure~\ref{fig-trans}). More precisely, we let
\begin{equation}\phantomsection\label{eq-defipcw}{
T^*_{\mathrm{IPCW}}=\frac{\Delta^\tau}{G(\widetilde{T}\wedge \tau|X,A)} \widetilde{T} \wedge \tau,
}\end{equation} where we recall that
\(G(t|X,A) :=\mathbb{P}(C \geqslant t|X,A)\) is the left limit of the
conditional survival function of the censoring. This estimator assigns
higher weights to uncensored subjects within a covariate group that is
highly prone to censoring, thereby correcting for conditionally
independent censoring and removing selection bias \citep{howe2016}.

\textcolor{black}{Using proposition~\ref{prp-cutest} and \ref{prp-cutestnopi}, we show in Appendix that the IPCW transformation is a censoring-unbiased transformation in the sense of Equation~\ref{eq-cut} (Proposition~\ref{prp-ipcw}).} 

The IPCW depends on the unknown conditional survival function of the
censoring \(G(\cdot|X,A)\), which thus needs to be estimated. Estimating
conditional censoring or the conditional survival function can be
approached similarly, as both involve estimating a time, whether for
survival or censoring. Consequently, we can use parametric methods, or semi-parametric methods,
such as the Cox model, or non-parametric approaches like survival
forests (see Section~\ref{sec-Gformula}). Once \(\widehat G(\cdot|A,X)\) is provided, one can construct an estimator of the difference of
RMST based on Equation~\ref{eq-cutestnopi},
\begin{equation}\phantomsection\label{eq-ipcwdefnopi}{
\widehat\theta_{\mathrm{IPCW}} = \frac1{n_1}\sum_{A_i = 1}  T^*_{\mathrm{IPCW},i} - \frac1{n_0}\sum_{A_i = 0}  T^*_{\mathrm{IPCW},i}.
}\end{equation} 

\textcolor{black}{Alternatively, instead of rescaling by $n_a$, the estimator can be stabilized by dividing by the sum of the weights, yielding a Hájek-type normalization that reduces variance without affecting consistency \citep{Hajek1971,Robins2000,Hernan2020}.}

By Proposition~\ref{prp-cutest} (\nameref{prp-cutest}), Proposition~\ref{prp-cutestnopi} (\nameref{prp-cutestnopi}) and
Proposition~\ref{prp-ipcw}  (\nameref{prp-ipcw}) (in Appendix~\ref{sec-proof22}), we easily deduce that
\(\widehat\theta_{\mathrm{IPCW}}\) is asymptotically consistent as soon
as \(\widehat G\) is.

\begin{corollary}[]\protect\hypertarget{cor-ipcwcons}{}\label{cor-ipcwcons}

Under
Assumptions~\ref{ass:sutva} (\nameref{ass:sutva}), \ref{ass:rta} (\nameref{ass:rta}), \ref{ass:trialpositivity} (\nameref{ass:trialpositivity}), \ref{ass:condindepcensoring} (\nameref{ass:condindepcensoring})
and \ref{ass:positivitycensoring} (\nameref{ass:positivitycensoring}), if \(\widehat G\) is uniformly weakly
(resp. strongly) consistent then so is
\(\widehat\theta_{\mathrm{IPCW}}\), defined in Equation~\ref{eq-ipcwdefnopi}.
\end{corollary}

This result simply comes from the fact that
\(\widehat\theta_{\mathrm{IPCW}}\) depends continuously on
\(\widehat G\) and that \(G\) is lower-bounded (Assumption~\ref{ass:positivitycensoring} \nameref{ass:positivitycensoring}). Surprisingly, we found limited use of this estimator in the literature, beside its first introduction in \cite{koul1981}.

A more common variant is the IPCW–Kaplan–Meier \citep{IPCWrobins, IPCW_consistency, IPCW}, defined as follows. In this approach, censoring weights are dynamically evaluated at each distinct event time, allowing the estimator to update risk sets and weighted failures throughout the follow-up period.

\begin{definition}[]\protect\hypertarget{def-ipcwkm}{}\label{def-ipcwkm}
(IPCW-Kaplan-Meier) We let again \(\widehat G(\cdot|X,A)\) be an
estimator of the (left limit of) the conditional censoring survival
function and we introduce

\begin{align*}
D_k^{\mathrm{IPCW}}(a) &:= \sum_{i=1}^n \frac{\Delta_i^\tau}{\widehat G(t_k| X_i,A=a)} \mathbb{I}(\widetilde T_i = t_k, A_i=a) \\
\quad\text{and}\quad N^{\mathrm{IPCW}}_k(a) &:= \sum_{i=1}^n \frac{1}{\widehat G(t_k  | X_i,A=a)} \mathbb{I}(\widetilde T_i \geqslant t_k, A_i=a),
\end{align*}
be the weight-corrected numbers of deaths and of individuals at risk at
time \(t_k\). The IPCW version of the KM estimator is defined as:
\[
\begin{aligned}
\widehat{S}_{\mathrm{IPCW}}(t | A=a) &= \prod_{t_k \leqslant t}\left(1-\frac{D_k^{\mathrm{IPCW}}(a)}{N_k^{\mathrm{IPCW}}(a)}\right). 
\end{aligned}
\]
\end{definition}
Note that the quantity \(\pi\) is not present in the definition of
\(D_k^{\mathrm{IPCW}}(a)\) and \(N_k^{\mathrm{IPCW}}(a)\) because it
would simply disappear in the ratio
\(D_k^{\mathrm{IPCW}}(a)/N_k^{\mathrm{IPCW}}(a)\). The subsequent RMST
estimator then takes the form
\begin{equation}\phantomsection\label{eq-thetaIPCWKM}{
\widehat{\theta}_{\mathrm{IPCW-KM}} = \int_{0}^{\tau}\widehat{S}_{\mathrm{IPCW}}(t|A=1)-\widehat{S}_{\mathrm{IPCW}}(t|A=0)dt.
}\end{equation} Like before for the classical KM estimator, this new
reweighted KM estimator enjoys good statistical properties.

\textcolor{black}{We show in Proposition~\ref{prp-ipcwkm} (Appendix~\ref{appendix-a-proofs}) that the oracle version of this estimator is strongly consistent and asymptotically normal. These properties also hold for the estimator using a consistently estimated $\widehat{G}$.}

\begin{corollary}[]\protect\hypertarget{cor-ipcwkmcons}{}\label{cor-ipcwkmcons}

Under Assumptions~\ref{ass:sutva} (\nameref{ass:sutva}), \ref{ass:rta} (\nameref{ass:rta}), \ref{ass:condindepcensoring} (\nameref{ass:condindepcensoring})
and \ref{ass:positivitycensoring} (\nameref{ass:positivitycensoring}), if \(\widehat G\) is uniformly weakly
(resp. strongly) consistent then so is
\(\widehat S_{\mathrm{IPCW}}(t|A=a)\).

\end{corollary}
Because the evaluation of \(N_k^{\textrm{IPCW}}(a)\) now depends on information gathered after time \(t_k\), the previous derivations of bias (Appendix~\ref{sec-proof21}, proposition~\ref{prp-km}) and asymptotic variance (Appendix~\ref{sec-proof21}, proposition~\ref{prp-varkm}) no longer hold. \textcolor{black}{\citet{Overgaard2025} derived asymptotic variance expressions for related Kaplan–Meier–based inverse probability-of-censoring weighted estimators using influence-function representations. While these results theoretically extend to our IPCW–Kaplan–Meier formulation, the resulting expressions remain difficult to interpret or implement in practice. Based on our own results, we can nevertheless establish that when \(G\) is unknown, consistency still holds provided that the estimator \(\widehat G(\cdot|A,X)\) is itself consistent.}

\textbf{The Buckley-James transformation}

The Buckley-James (BJ) transformation, introduced by \cite{buckley1979}, takes a different path by leaving all uncensored values
untouched, while replacing the censored ones by an extrapolated value (see Figure~\ref{fig-trans}).
Formally, it is defined as follows:
\begin{equation}\phantomsection\label{eq-defbj}{
\begin{aligned}
T^*_{\mathrm{BJ}} &= \Delta^\tau (\widetilde{T}\wedge\tau) + (1-\Delta^\tau) Q_S(\widetilde T \wedge \tau|X,A),
\end{aligned}
}\end{equation}

where, for \(t \leqslant\tau\),
\[Q_S(t|X,A) :=\mathbb{E}[T \wedge \tau | X,A,T \wedge \tau > t]= \frac{1}{S(t|X,A)}\int_{t}^{\tau} S(u|X,A) \mathop{}\!\mathrm{d}u\]
and \(S(t|X,A=a) := \mathbb{P}(T^ {(a)} > t|X)\) is the conditional
survival function.

\textcolor{black}{This transformation satisfies the censoring-unbiasedness condition defined in Equation~\ref{eq-cut} (Proposition~\ref{prp-bj} in Appendix~\ref{appendix-a-proofs}).}
The BJ transformation depends on a
nuisance parameter (\(Q_S(\cdot|X,A)\)) that needs to be estimated. Estimation strategies described in Section~\ref{sec-Gformula} can be used. Once provided with an estimator \(\widehat Q_S(\cdot|A,X)\), a natural estimator of the
RMST based on the BJ transformation and Equation~\ref{eq-cutestnopi} is \begin{equation}\phantomsection\label{eq-BJnopi}{
\widehat\theta_{\mathrm{BJ}} = \frac1{n_1} \sum_{A_i=1}  T^*_{\mathrm{BJ},i}-\frac1{n_0} \sum_{A_i=0}  T^*_{\mathrm{BJ},i}. 
}\end{equation}

Like for the IPCW transformation, the BJ transformation
yields a consistent estimate of the RMST as soon as the model is
well-specified.

\begin{corollary}[]\protect\hypertarget{cor-bjcons}{}\label{cor-bjcons}

Under Assumptions~\ref{ass:sutva} (\nameref{ass:sutva}), \ref{ass:rta} (\nameref{ass:rta}), \ref{ass:condindepcensoring} (\nameref{ass:condindepcensoring})
and \ref{ass:positivitycensoring} (\nameref{ass:positivitycensoring}), if \(\widehat Q_S\) is uniformly weakly
(resp. strongly) consistent then so is \(\widehat\theta_{\mathrm{BJ}}\)
defined as in Equation~\ref{eq-BJnopi}.

\end{corollary}

The proof is an application of Propositions
\ref{prp-cutest} (\nameref{prp-cutest}), \ref{prp-cutestnopi} (\nameref{prp-cutestnopi}) and \ref{prp-bj} (\nameref{prp-bj} in Appendix, Section~\ref{sec-proof22}), and relies on
the continuity of \(S \mapsto Q_S\). \textcolor{black}{Additionally, the Buckley-James transformation enjoys a notable optimality property.}

\begin{theorem}[BJ minimizes the mean squared error among CUT]\protect\hypertarget{thm-bj}{}\label{thm-bj}

For any transformation \(T^*\) of the form \ref{eq-defcut}, it holds \[
\mathbb{E}[(T^*_{\mathrm{BJ}}-T \wedge \tau)^2] \leqslant\mathbb{E}[(T^*-T \wedge \tau)^2].
\]

\end{theorem}

This result, originally stated without proof in \cite{fan1994}, is proved there for completeness in Appendix, Section~\ref{sec-proof22}.
\subsubsection{Augmented corrections}\label{sec-tdr}

The main disadvantage of the two previous transformations is that they
heavily rely on the specification of good estimator \(\widehat G\) (for
IPCW) or \(\widehat S\) (for BJ). In order to circumvent this
limitation, \cite{rubin2007} proposed the following
transformations, whose expression is based on theory of semi-parametric
estimation developed in \cite{vanderLaan2003},\\
\begin{equation}\phantomsection\label{eq-TDR}{
T^*_\mathrm{DR} = \frac{\Delta^\tau \widetilde T\wedge \tau}{G(\widetilde T \wedge \tau|X,A)} + \frac{(1-\Delta^\tau)Q_S(\widetilde T \wedge \tau |X,A)}{G(\widetilde T \wedge \tau |X,A)}- \int_0^{\widetilde T \wedge \tau} \frac{Q_S(t|X,A)}{G(t|X,A)^2} \mathop{}\!\mathrm{d}\mathbb{P}_C(t|X,A),
}\end{equation} where \(\mathop{}\!\mathrm{d}\mathbb{P}_C(t|X,A)\) is
the distribution of \(C\) conditional on the covariates \(X\) and \(A\).
As this transformation depends on both conditional survival functions \(G\) and \(S\), we will denote it \(T^*_\mathrm{DR}(G,S)\) to make this dependence explicit.
\textcolor{black}{Importantly, this transformation not only satisfies the censoring-unbiasedness condition of Equation~\ref{eq-cut}, but also enjoys a double robustness property: it remains valid as long as either \(G\) or \(S\) is correctly specified. The formal statement of this result is given in Proposition~\ref{prp-tdr} in Appendix~\ref{appendix-a-proofs}. The statement and proof of this results is found in \cite{rubin2007} in the case where \(C\) and \(T\) are continuous.}

\section{Causal survival analysis in observational studies}\label{sec-theoryOBS}

Unlike RCT, observational data (such as from registries, electronic
health records, or national healthcare systems) are collected without
controlled randomized treatment allocation. In such cases, treated and
control groups are likely unbalanced due to the non-randomized design,
which results in the treatment effect being confounded by variables
influencing both the survival outcome \(T\) and the treatment allocation
\(A\). To enable identifiability of the causal effect, additional
standard assumptions are needed.

\begin{assumption}
\phantomsection
\assumptionname{Conditional Exchangeability / Unconfoundedness}
\label{ass:unconf}
\textbf{(Conditional Exchangeability / Unconfoundedness)}
It holds
$A \perp\mkern-9.5mu\perp T^ {(0)}, T^ {(1)} \mid X$
\end{assumption}

Under Assumption~\ref{ass:unconf}, the treatment assignment is randomly
assigned conditionally on the covariates \(X\). This assumption states
that there are no unmeasured confounders as the latter would make it
impossible to distinguish correlation from causality.

\begin{assumption}
\phantomsection
\assumptionname{Positivity / Overlap for Treatment}
\label{ass:positivitytreat}
\textbf{(Positivity / Overlap for Treatment)}
Letting \(e(X) := \mathbb{P}(A=1\mid X)\) be the \emph{propensity score}, there holds $0 < e(X) < 1 \quad \text{almost surely.}$
\end{assumption}

The Assumption~\ref{ass:positivitytreat} requires adequate
overlap in covariate distributions between treatment groups, meaning
every observation must have a non-zero probability of being treated and not being treated.
Because Assumption \ref{ass:rta} (\nameref{ass:rta}) no longer holds anymore, neither does the
previous identifiability Equation~\ref{eq-RMSTkm}. In this new context, we
can write
\begin{equation}\phantomsection\label{eq-identcond}{
\begin{aligned}
    \theta_{\mathrm{RMST}} &=  \mathbb{E}[T^ {(1)} \wedge \tau - T^ {(0)} \wedge \tau]\\
    &= \mathbb{E}\left[\mathbb{E}[T^ {(1)} \wedge \tau|X] - \mathbb{E}[T^ {(0)} \wedge \tau|X] \right]  &&  \\
    &=\mathbb{E}\left[\mathbb{E}[T^ {(1)} \wedge \tau|X, A=1] - \mathbb{E}[T^ {(0)} \wedge \tau|X, A=0]\right]   && \text{\small(Unconfoundedness)} \\
       &= \mathbb{E}\left[\mathbb{E}[T \wedge \tau|X, A=1] - \mathbb{E}[T \wedge \tau|X, A=0]\right].  && \text{\small(SUTVA)}
\end{aligned}
}\end{equation} 

A direct application of the above equation could be a direct plug-in G-formula estimator developed in
Section~\ref{sec-Gformula}, equation~\ref{eq-gformula}.
In particular, \cite{chen2001} prove consistency and asymptotic normality
results for Cox estimators in an observational study, and they give an
explicit formulation of the asymptotic variance as a function of the
parameters of the Cox model. In the non-parametric literature, \cite{foster2011} and \cite{kunzel2019} empirically study this estimator using survival forests, with the former employing it as a
T-learner (referred to as \emph{Virtual Twins}) and the latter as an
S-learner.

In another direction, the next sections describe alternative estimation strategies that also rely on equation~\ref{eq-identcond} but adapt the estimation of the conditional expectations to the specific censoring assumptions—either independent or conditionally independent censoring.

\subsection{Independent censoring}\label{sec-obs_indcen}

Under Assumptions~\ref{ass:unconf} (\nameref{ass:unconf}) and \ref{ass:independentcensoring} (\nameref{ass:independentcensoring}), the data exhibit confounding effect induced by the confounders. The corresponding causal graph is shown in Figure~\ref{fig-causalgraph_obs_ind}.

\begin{figure}[h]

\centering{

\includegraphics[width=0.35\textwidth]{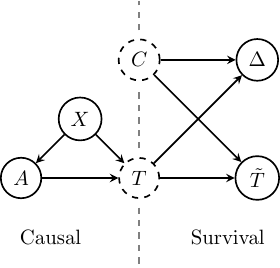}

}

\caption{\label{fig-causalgraph_obs_ind}Causal graph in observational
survival data with independent censoring.}

\end{figure}%

Under Assumption~\ref{ass:independentcensoring} (\nameref{ass:independentcensoring}), we saw in
Section~\ref{sec-theoryRCT_indc} that the Kaplan-Meier estimator could
conveniently handle censoring. Building on Equation~\ref{eq-identcond}, we
can write
\begin{equation}\phantomsection\label{eq-identcond2}{
\begin{aligned}
    \theta_{\mathrm{RMST}} &= \mathbb{E}\left[\mathbb{E}[T \wedge \tau|X, A=1] - \mathbb{E}[T \wedge \tau|X, A=0]\right] \\
    &= \mathbb{E}\left[\frac{\mathbb{E}[\mathbb{I}\{A=1\} T \wedge \tau |X]}{\mathbb{E}[\mathbb{I}\{A=1\}|X]}\right] - \mathbb{E}\left[\frac{\mathbb{E}[\mathbb{I}\{A=0\} T \wedge \tau |X]}{\mathbb{E}[\mathbb{I}\{A=0\}|X]}\right] && \text{\small(Positivity for treatment)} \\
    &= \mathbb{E}\left[\frac{A T \wedge \tau}{e(X)}\right] - \mathbb{E}\left[\frac{(1-A) T \wedge \tau}{1-e(X)}\right],
\end{aligned}
}\end{equation} 
which suggests to adapt the classical KM estimator by reweighting it
by the propensity score (see equation~\ref{eq-rmstsurv}). The use of propensity score in causal inference
has been initially introduced by \cite{rosenbaum1983} and further
developed in \cite{hirano2003}. It was extended to
survival analysis by \cite{xie2005} through the adjusted
Kaplan-Meier estimator as defined below.

\begin{definition}[]\protect\hypertarget{def-iptwkm}{}\label{def-iptwkm}

(IPTW Kaplan-Meier estimator) We let \(\widehat e(\cdot)\) be an
estimator of the propensity score \(e(\cdot)\). We introduce
\begin{align*}
D_k^{\mathrm{IPTW}}(a) &:= \sum_{i=1}^n \left(\frac{A_i}{\widehat e(X_i)}+\frac{1-A_i}{1- \widehat e(X_i)}\right)\mathbb{I}(\widetilde T_i = t_k, \Delta_i = 1, A_i=a) \\
\quad\text{and}\quad N^{\mathrm{IPTW}}_k(a) &:= \sum_{i=1}^n \left(\frac{A_i}{\widehat e(X_i)}+\frac{1-A_i}{1- \widehat e(X_i)}\right) \mathbb{I}(\widetilde T_i \geqslant t_k, A_i=a),
\end{align*}
be the numbers of deaths and of individuals at risk at time \(t_k\),
reweighted by the propensity score. The Inverse-Probability-of-Treatment
Weighting (IPTW) version of the KM estimator is defined as:
\begin{equation}\phantomsection\label{eq-IPTWKM}{
\begin{aligned}
\widehat{S}_{\mathrm{IPTW}}(t | A=a) &= \prod_{t_k \leqslant t}\left(1-\frac{D_k^{\mathrm{IPTW}}(a)}{N_k^{\mathrm{IPTW}}(a)}\right). 
\end{aligned}
}\end{equation}
\end{definition}
We let \(S^*_{\mathrm{IPTW}}(t | A=a)\) be the oracle KM-estimator
defined as above with \(\widehat e(\cdot) = e(\cdot)\). Although the
reweighting slightly changes the analysis, the good properties of the
classical KM carry on to this setting.

\textcolor{black}{In Appendix~\ref{appendix-a-proofs}, Proposition~\ref{prp-iptwkm} establishes that \(S^*_{\mathrm{IPTW}}(t \mid A=a)\) is a strongly consistent and asymptotically normal estimator of the survival function \(S^{(a)}(t)\). Corollary~\ref{cor-iptwkm} further shows that the same properties hold for the practical estimator provided the estimated propensity score converges uniformly.}
This follow directly from \(S^*_{\mathrm{IPTW}}\) is a continuous function of
\(e(\cdot)\), and because \(e\) and \(1-e\) are lower-bounded as per
Assumptions~\ref{ass:positivitytreat} (\nameref{ass:positivitytreat}).

\begin{corollary}[]\protect\hypertarget{cor-iptwkm}{}\label{cor-iptwkm}

Under Assumptions
\ref{ass:sutva} (\nameref{ass:sutva}),  \ref{ass:unconf} (\nameref{ass:unconf}), \ref{ass:positivitytreat} (\nameref{ass:positivitytreat}), \ref{ass:independentcensoring} (\nameref{ass:independentcensoring}) and \ref{ass:poscen} (\nameref{ass:poscen}), if
\(\widehat e(\cdot)\) satisfies \(\|\widehat e-e\|_{\infty} \to 0\)
almost surely (resp. in probability), then the IPTW Kaplan-Meier
estimator \(\hat S_{\mathrm{IPTW}}(t | A=a)\) is a strongly (resp.
weakly) consistent estimator of \(S^{(a)}(t)\).
\end{corollary}

The resulting RMST estimator simply takes the form:
\begin{equation}\phantomsection\label{eq-RMST_IPTWKM}{
\widehat{\theta}_{\mathrm{IPTW-KM}} = \int_{0}^{\tau}\widehat{S}_{\mathrm{IPTW}}(t|A=1)-\widehat{S}_{\mathrm{IPTW}}(t|A=0)dt.
}\end{equation} which will be consistent under the same Assumptions as
the previous Corollary. Note that we are not aware of any formal
results concerning the bias and the asymptotic variance of the oracle
estimator \(S^*_{\mathrm{IPTW}}(t | A=a)\), and we refer to Xie and Liu
(2005) for heuristics concerning these questions.

\subsection{Conditional independent censoring}\label{sec-obscondcens}

Under Assumptions~\ref{ass:unconf} (\nameref{ass:unconf}) and \ref{ass:condindepcensoring} (\nameref{ass:condindepcensoring}), the data exhibit both confounding and conditionally independent censoring. The corresponding causal graph is shown in Figure~\ref{fig-causalgraph_obs_dep}. To address this, weight the IPCW and Buckley–James (BJ) transformations from Section~\ref{sec-unbiasedtransfo} by the propensity score.
\begin{figure}[h]
\centering{
\includegraphics[width=0.35\textwidth]{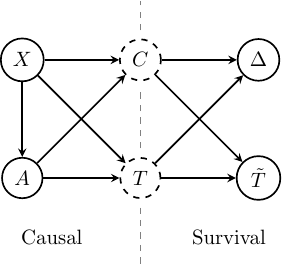}
}
\caption{\label{fig-causalgraph_obs_dep}Causal graph in observational
survival data with dependent censoring.}
\end{figure}%
\subsubsection{IPTW-IPCW transformations}\label{iptw-ipcw-transformations}

One can check that the IPCW transformation as introduced in
Equation~\ref{eq-defipcw} is also a censoring-unbiased transformation in
observational study (see Appendix in Proposition~\ref{prp-iptwipcw}).
Deriving an estimator of the difference of
RMST is however slightly different in that context. In particular,
Equation~\ref{eq-cutcond} rewrites \[
\mathbb{E}[\mathbb{E}[T^*|X,A=1]] = \mathbb{E}\left[\frac{A}{e(X)} T^*\right],
\] which suggests defining an estimator of the RMST difference \textcolor{black}{that can, similarly to IPCW estimators, be stabilized by normalizing with the sum of the weights,} and is identified by combining Equation~\ref{eq-identcond2} with the IPCW representation in Equation~\ref{eq-defipcw}.
\begin{equation}\phantomsection\label{eq-iptwipcw}{
\widehat\theta_{\mathrm{IPTW-IPCW}} = \frac1n\sum_{i=1}^n  \left(\frac{A_i}{e(X_i)}-\frac{1-A_i}{1-e(X_i)} \right) T^*_{\mathrm{IPCW},i}.
}\end{equation}
This transformation now depends on two nuisance parameters, namely the
conditional survival function of the censoring (through
\(T^*_{\mathrm{IPCW}}\)) and the propensity score. \textcolor{black}{Proposition~\ref{prp-consiptwipcw} ensures consistency of this estimator when both nuisance quantities are consistently estimated.}

\begin{proposition}[Consistency of $\widehat\theta_{\mathrm{IPTW-IPCW}}$ with estimated $\hat{e}$ and $\hat{G}$]\protect\hypertarget{prp-consiptwipcw}{}\label{prp-consiptwipcw}

Under Assumptions
\ref{ass:sutva} (\nameref{ass:sutva}),  \ref{ass:unconf} (\nameref{ass:unconf}), \ref{ass:positivitytreat} (\nameref{ass:positivitytreat}), \ref{ass:condindepcensoring} (\nameref{ass:condindepcensoring})
and \ref{ass:positivitycensoring} (\nameref{ass:positivitycensoring}), and if \(\widehat G(\cdot|X,A)\) and
\(\widehat e (\cdot)\) are uniformly weakly (resp. strongly) consistent
estimators, then estimator \ref{eq-iptwipcw} defined with \(\widehat e\)
and \(\widehat G\) is a weakly (resp. strongly) consistent estimator of
the RMST difference.

\end{proposition}

\textcolor{black}{We can also use the same strategy as for the
IPCW transformation and incorporate the time-dependent weights into a Kaplan-Meier estimator.}

\begin{definition}[]\protect\hypertarget{def-iptwipcwkm}{}\label{def-iptwipcwkm}

(IPTW-IPCW-Kaplan-Meier) We let again \(\widehat G(\cdot|X,A)\) and
\(\widehat e(\cdot)\) be estimators of the conditional censoring
survival function and of the propensity score. We introduce

\begin{align*}
D_k^{\mathrm{IPTW-IPCW}}(a) &:= \sum_{i=1}^n \left(\frac{A_i}{\widehat e(X_i)}+\frac{1-A_i}{1-\widehat e(X_i)} \right)\frac{\Delta_i^\tau}{\widehat G(t_k | X_i,A=a)} \mathbb{I}(\widetilde T_i = t_k, A_i=a) \\
\quad\text{and}\quad N^{\mathrm{IPTW-IPCW}}_k(a) &:= \sum_{i=1}^n \left(\frac{A_i}{\widehat e(X_i)}+\frac{1-A_i}{1-\widehat e(X_i)} \right)\frac{1}{\widehat G(t_k | X_i,A=a)} \mathbb{I}(\widetilde T_i \geqslant t_k, A_i=a),
\end{align*}

\textcolor{black}{be the weight-corrected numbers of deaths and of individual at risk at
time \(t_k\).} The IPTW-IPCW version of the KM estimator is defined
as: \[
\begin{aligned}
\widehat{S}_{\mathrm{IPTW-IPCW}}(t | A=a) &= \prod_{t_k \leqslant t}\left(1-\frac{D_k^{\mathrm{IPTW-IPCW}}(a)}{N_k^{\mathrm{IPTW-IPCW}}(a)}\right). 
\end{aligned}
\]

\end{definition}

The RMST difference estimated with IPTW-IPCW-Kaplan-Meier survival
curves is then simply as

\begin{equation}\phantomsection\label{eq-RMST_IPTW_IPCWKM}{
\widehat{\theta}_{\mathrm{IPTW-IPCW-KM}} = \int_{0}^{\tau}\widehat{S}_{\mathrm{IPTW-IPCW}}(t|A=1)-\widehat{S}_{\mathrm{IPTW-IPCW}}(t|A=0)dt.
}\end{equation}

\textcolor{black}{In Appendix~\ref{appendix-a-proofs}, Proposition~\ref{prp-iptwipcwkm} shows that the oracle version of this estimator is consistent and asymptotically normal. The result extends to the empirical estimator: as shown in Corollary~\ref{cor-consiptwipcwkm}, uniform consistency of the nuisance estimators ensures the consistency of the estimated survival curves
\(\widehat{\theta}_{\mathrm{IPTW-IPCW-KM}}\).}

\begin{corollary}[]\protect\hypertarget{cor-consiptwipcwkm}{}\label{cor-consiptwipcwkm}

Under Assumptions
\ref{ass:sutva} (\nameref{ass:sutva}),  \ref{ass:unconf} (\nameref{ass:unconf}), \ref{ass:positivitytreat} (\nameref{ass:positivitytreat}), \ref{ass:condindepcensoring} (\nameref{ass:condindepcensoring})
and \ref{ass:positivitycensoring} (\nameref{ass:positivitycensoring}), and for all \(t \in [0,\tau]\), if the
nuisance estimators \(\widehat G(\cdot|A,X)\) and \(\widehat e\) are
weakly (resp. strongly) uniformly consistent, then
\(\widehat{S}_{\mathrm{IPTW-IPCW}}(t | A=a)\) is a weakly (resp.
strongly) consistent estimator of \(S^{(a)}(t)\).

\end{corollary}

We are not aware of general formulas for the asymptotic variances in this
context. We mention nonetheless that \cite{schaubel2011} have been
able to derive asymptotic laws in this framework in the particular case
of Cox-models.

\subsubsection{IPTW-BJ transformations}\label{iptw-bj-transformations}

As with IPCW, the BJ transformation in Equation~\ref{eq-defbj} is censoring-unbiased in observational contexts (Proposition~\ref{prp-iptwbj}, Appendix~\ref{appendix-a-proofs}).
The corresponding estimator of the difference
in RMST is \begin{equation}\phantomsection\label{eq-iptwbj}{
\hat \theta_{\mathrm{IPTW-BJ}} = \frac1n\sum_{i=1}^n  \left(\frac{A_i}{e(X_i)}-\frac{1-A_i}{1-e(X_i)} \right) T^*_{\mathrm{BJ},i}.
}\end{equation}

This transformation depends on the conditional survival function \(S\)
(through \(T^*_{\mathrm{BJ}}\)) and the propensity score. Proposition \ref{prp-consiptwbj} stated that consistency of
the nuisance parameter estimators implies again consistency of the RMST
estimator. 

\begin{proposition}[Consistency of $\widehat\theta_{\mathrm{IPTW-BJ}}$ with estimated $\hat{e}$ and $\hat{S}$]\protect\hypertarget{prp-consiptwbj}{}\label{prp-consiptwbj}

Under Assumptions
\ref{ass:sutva} (\nameref{ass:sutva}),  \ref{ass:unconf} (\nameref{ass:unconf}), \ref{ass:positivitytreat} (\nameref{ass:positivitytreat}), \ref{ass:condindepcensoring} (\nameref{ass:condindepcensoring})
and \ref{ass:positivitycensoring} (\nameref{ass:positivitycensoring}), and if \(\widehat S(\cdot|X,A)\) and
\(\widehat e (\cdot)\) are uniformly weakly (resp. strongly) consistent
estimators, then \(\widehat\theta_{\mathrm{IPTW-BJ}}\) defined with
\(\widehat S\) and \(\widehat e\) is a weakly (resp. strongly)
consistent estimator of the RMST.

\end{proposition}

\subsubsection{Double augmented corrections}\label{sec-AIPTW_AIPCW}

Building on Equations~\ref{eq-identcond} and \ref{eq-identcond2}, it follows that
\[
\mathbb{E}\!\left[
\frac{\mathbb{I}\{A=a\}}{\mathbb{E}[\mathbb{I}\{A=a\}\mid X]}\,
\Big( (T\wedge\tau)-\mathbb{E}[T\wedge\tau\mid X,A=a] \Big)
\,\middle|\, X
\right]=0.
\]
Thus,
\begin{equation}
\begin{aligned}
\theta_{\mathrm{RMST}}
&= \mathbb{E}\left[\mathbb{E}[T\wedge\tau\mid X,A=1]-\mathbb{E}[T\wedge\tau\mid X,A=0]\right] \\
&= \mathbb{E}\left[
\mathbb{E}[T\wedge\tau\mid X,A=1]-\mathbb{E}[T\wedge\tau\mid X,A=0]
\right. \\
&\qquad \left.
+ \underbrace{
\frac{\mathbb{I}\{A=1\}}{e(X)}\Big(T\wedge\tau-\mathbb{E}[T\wedge\tau\mid X,A=1]\Big)
- \frac{\mathbb{I}\{A=0\}}{1-e(X)}\Big(T\wedge\tau-\mathbb{E}[T\wedge\tau\mid X,A=0]\Big)
}_{\text{\small add mean-zero terms}}
\right].
\end{aligned}
\end{equation}

This representation coincides with the classical doubly robust augmented inverse probability weighting (AIPTW) functional from causal inference \citep{robins1994}. Incorporating the doubly robust censoring transformations from Section~\ref{sec-tdr} yields a transformation that is doubly robust with respect to both censoring and treatment assignment\[
\Delta^*_{\mathrm{QR}} = \Delta^*_{\mathrm{QR}}(G,S,\mu,e)  := \left(\frac{A}{e(X)}-\frac{1-A}{1-e(X)}\right)(T^*_{\mathrm{DR}}(G,S)-\mu(X,A))+\mu(X,1)-\mu(X,0),
\] where we recall that \(T^*_{\mathrm{DR}}\) is defined in
Section~\ref{sec-tdr}. This transformation depends on four nuisance
parameters: \(G\) and \(S\) through \(T^*_{\mathrm{DR}}\), and now the
propensity score \(e\) and the conditional response \(\mu\). This
transformation does not really fall into the scope of censoring-unbiased
transform, but it is easy to show that \(\Delta^*_{\mathrm{QR}}\) is
doubly robust for censoring and treatment assignment. \textcolor{black}{That is, among the four nuisance functions involved, only one from each pair \((S,G)\) and \((e,\mu)\) needs to be correctly specified. This property (similar to \cite{Ozenne20_AIPTW_AIPCW}, Thm 1) is formally stated in Proposition~\ref{prp-tqr} in Appendix~\ref{appendix-a-proofs}.}

Based on estimators
\((\widehat G, \widehat S, \widehat\mu, \widehat e)\) of
\((G,S,\mu,e)\), one can then propose the following estimator of the
RMST, coined the AIPTW-AIPCW estimator in \cite{Ozenne20_AIPTW_AIPCW}:
\begin{equation}\phantomsection\label{eq-AIPTW_AIPCW}{
\begin{aligned}
\widehat\theta_{\mathrm{AIPTW-AIPCW}} &:= \frac1n \sum_{i=1}^n \Delta_{\mathrm{QR},i}^*(\widehat G, \widehat S, \widehat\mu, \widehat e)
\\
&=\frac1n \sum_{i=1}^n \left(\frac{A_i}{\widehat e(X_i)}-\frac{1-A_i}{1-\widehat e(X_i)}\right)(T^*_{\mathrm{DR}}(\widehat G,\widehat S)_i-\widehat\mu(X_i,A_i)) + \widehat\mu(X_i,1)-\widehat\mu(X_i,0).
\end{aligned}
}\end{equation}

Asymptotic properties of this estimator, such as asymptotic normality and semiparametric efficiency, are detailed in \cite{Ozenne20_AIPTW_AIPCW}.

\section{Summary of the estimators and variance estimation}

For completeness, Table~\ref{tbl-nuisance} offers a compact synopsis of the estimators, context of application, associated nuisance models, and large-sample behavior, with a focus on misspecification sensitivity.
\hspace*{-0.9cm}
\begin{longtable}[]{@{}
  >{\raggedright\arraybackslash}p{0.23\linewidth}
  >{\raggedright\arraybackslash}p{0.03\linewidth}
  >{\raggedright\arraybackslash}p{0.03\linewidth}
  >{\raggedright\arraybackslash}p{0.04\linewidth}
  >{\raggedright\arraybackslash}p{0.04\linewidth}
  >{\raggedright\arraybackslash}p{0.1\linewidth}
  >{\raggedright\arraybackslash}p{0.1\linewidth}
  >{\raggedright\arraybackslash}p{0.1\linewidth}
  >{\raggedright\arraybackslash}p{0.13\linewidth}@{}}
\caption{Estimators of the RMST difference and nuisance parameters
needed to compute each estimator. Empty boxes indicate that the nuisance
parameter is not needed in the estimator thus misspecification has no effect.}
\label{tbl-nuisance}\\
\toprule\noalign{}
\begin{minipage}[b]{\linewidth}\raggedright
Estimator
\end{minipage} & \begin{minipage}[b]{\linewidth}\raggedright
RCT
\end{minipage} & \begin{minipage}[b]{\linewidth}\raggedright
Obs
\end{minipage} & \begin{minipage}[b]{\linewidth}\raggedright
Ind Cens
\end{minipage} & \begin{minipage}[b]{\linewidth}\raggedright
Dep Cens
\end{minipage} & \begin{minipage}[b]{\linewidth}\raggedright
Outcome model
\end{minipage} & \begin{minipage}[b]{\linewidth}\raggedright
Censoring model
\end{minipage} & \begin{minipage}[b]{\linewidth}\raggedright
Treatment model
\end{minipage} & \begin{minipage}[b]{\linewidth}\raggedright
Robustness
\end{minipage} \\
\midrule\noalign{}
\endfirsthead
\toprule\noalign{}
\begin{minipage}[b]{\linewidth}\raggedright
Estimator
\end{minipage} & \begin{minipage}[b]{\linewidth}\raggedright
RCT
\end{minipage} & \begin{minipage}[b]{\linewidth}\raggedright
Obs
\end{minipage} & \begin{minipage}[b]{\linewidth}\raggedright
Ind Cens
\end{minipage} & \begin{minipage}[b]{\linewidth}\raggedright
Dep Cens
\end{minipage} & \begin{minipage}[b]{\linewidth}\raggedright
Outcome model
\end{minipage} & \begin{minipage}[b]{\linewidth}\raggedright
Censoring model
\end{minipage} & \begin{minipage}[b]{\linewidth}\raggedright
Treatment model
\end{minipage} & \begin{minipage}[b]{\linewidth}\raggedright
Robustness
\end{minipage} \\
\midrule\noalign{}
\endhead
\bottomrule\noalign{}
\endlastfoot
\emph{Unadjusted KM} & \(\color{green}\checkmark\) & & \(\color{green} \checkmark\) & & & & & \\
IPCW-KM \& IPCW & \(\color{green} \checkmark\) & & \(\color{green} \checkmark\) & \(\color{green} \checkmark\) & & \(G\) & & \\
BJ & \(\color{green} \checkmark\) & & \(\color{green} \checkmark\) & \(\color{green} \checkmark\) & \(S\) & & & \\
\emph{IPTW-KM} & \(\color{green} \checkmark\) & \(\color{green} \checkmark\) & \(\color{green} \checkmark\) & & & & \(e\) & \\
IPCW-IPTW-KM \& IPCW-IPTW & \(\color{green} \checkmark\) & \(\color{green} \checkmark\) & \(\color{green} \checkmark\) & \(\color{green} \checkmark\) & & \(G\) & \(e\) & \\
IPTW-BJ & \(\color{green} \checkmark\) & \(\color{green} \checkmark\) & \(\color{green} \checkmark\) & \(\color{green} \checkmark\) & \(S\) & & \(e\) & \\
\emph{G-formula} & \(\color{green} \checkmark\) & \(\color{green} \checkmark\) & \(\color{green} \checkmark\) & \(\color{green} \checkmark\) & \(\mu\) & & & \\
\textcolor{black}{AIPCW} & \(\color{green} \checkmark\) & & \(\color{green} \checkmark\) & \({\color{green} \checkmark}\) & \(S\) & \(G\) &  & \(\color{green} \checkmark\) (Prp \ref{prp-tdr}) \\
AIPTW-AIPCW & \(\color{green} \checkmark\) & \(\color{green} \checkmark\) & \(\color{green} \checkmark\) & \({\color{green} \checkmark}\) & \(S,\mu\) & \(G\) & \(e\) & \(\color{green} \checkmark\) (Prp \ref{prp-tqr}) \\
\end{longtable}

\textcolor{black}{Some survival estimators presented in this work have known asymptotic variance expressions, such as the Kaplan-Meier estimator (via Greenwood's formula, derived in Proposition~\ref{prp-varkm}) or AIPTW-AIPCW estimator (derived in \cite{Ozenne20_AIPTW_AIPCW}). Under standard regularity conditions ensuring asymptotic normality and assuming independence between the two survival curves, the variance of the RMST difference can be readily obtained via the delta method.}

\textcolor{black}{However, in practice the two estimated survival curves are typically correlated, which complicates variance estimation. While explicit variance formulas have been derived in certain settings (such as for the G-formula combined with Cox estimators (see \cite{chen2001})) these expressions often remain complex and difficult to implement. We therefore recommend the nonparametric bootstrap \citep{efron1994} as a robust and practical solution for variance estimation and confidence interval construction in general settings. The effectiveness of bootstrap methods in the context of censored data has been investigated in several studies, including \cite{efron1981} and \cite{chen1996}, which provide both theoretical justification and empirical validation for its use with Kaplan–Meier estimator.}

\section{Simulations}\label{sec-simulation}

\subsection{Implementation}\label{implementation}
\textcolor{black}{Maintained implementations for estimating RMST under right censoring remain limited. Available tools include \href{https://cran.r-project.org/web/packages/survRM2/index.html}{\texttt{survRM2}}
 \citep{Hajime2015}, \href{https://cran.r-project.org/web/packages/grf/index.html}{\texttt{grf}}
 \citep{Tibshirani2017}, and \href{https://cran.r-project.org/web/packages/RISCA/index.html}{\texttt{RISCA}}
 \citep{Foucher2019}. The Appendix details these packages and their functions, mapping each to the estimators studied in this paper. We implement the remaining estimators and, for a unified and pedagogical framework, provide custom implementations for all estimators, including those with existing software. The complete codebase is available on GitHub at \href{https://github.com/Sanofi-Public/causal_survival_analysis}{\texttt{\detokenize{github.com/Sanofi-Public/causal_survival_analysis}}}
, with function-level usage documented in the Appendix~\ref{appendix-b-code}.}

Building on these implementations, we compare the behaviors and performances of the estimators through
simulations conducted across various experimental contexts.
These contexts include scenarios based on RCTs and observational data, with both independent and dependent censoring (see Appendix~\ref{appendix-c-figures} for details about the data generating process). \textcolor{black}{ For each setting, we generated 100 independent datasets, computed the estimator on each, and summarized the resulting distribution using boxplots.}

\subsection{RCT}\label{sec-simulation-RCT}

We sample \(n\) i.i.d. observations \((X_i,A_i,C_i,T_i^{(0)},T_i^{(1)})_{i\in[n]}\) where covariates
\(X\sim \mathcal N(\mu,\Sigma)\), treatment is randomized with \(e(X)=\Pr(A=1\mid X)=0.5\), and the potential event times follow Cox models. Specifically, the control potential time has hazard $\lambda^{(0)}(t\mid X)=\lambda_0\,\exp\{\beta_0^\top X\},$ and the treated potential time is obtained via a constant additive shift, $T^{(1)} = T^{(0)} + 10.$ Censoring is either independent, with a constant hazard \(\lambda_C(t)=\lambda_{C,0}\) (Scenario 1), or conditionally independent given \(X\), with a Cox-type hazard $\lambda_C(t\mid X)=\lambda_{C,0}\,\exp\{\beta_C^\top X\}$ (Scenario 2). The truncation horizon is fixed at \(\tau=25\). 
For RCT scenarios 1 and 2, the ground truth value of the RMST difference at time $\tau=25$ was computed analytically and found to be approximately 7.1.

For each scenario, we estimate the RMST difference using the methods
summarized in Table~\ref{tbl-nuisance}. The methods used to estimate
the nuisance components are indicated in brackets: either logistic
regression or random forests for propensity scores and either Cox models
or survival random forests for survival and censoring models. A naive
estimator where censored observations are simply removed and the
survival time is averaged for treated and controls is also provided for
a naive baseline.

Figure~\ref{fig-rct1} shows the distribution of the RMST difference
for 100 simulations in Scenario 1 and different sample sizes: 500, 1000,
2000, 4000. The true value of \(\theta_{\mathrm{RMST}}\) is indicated by
red dashed line.  \textcolor{black}{For sake of clarity, only a subset of estimators with particularly informative or interesting results is displayed.}
\begin{figure}[h]

\centering{
\hspace{-0.5cm}
\includegraphics[width=0.9\textwidth]{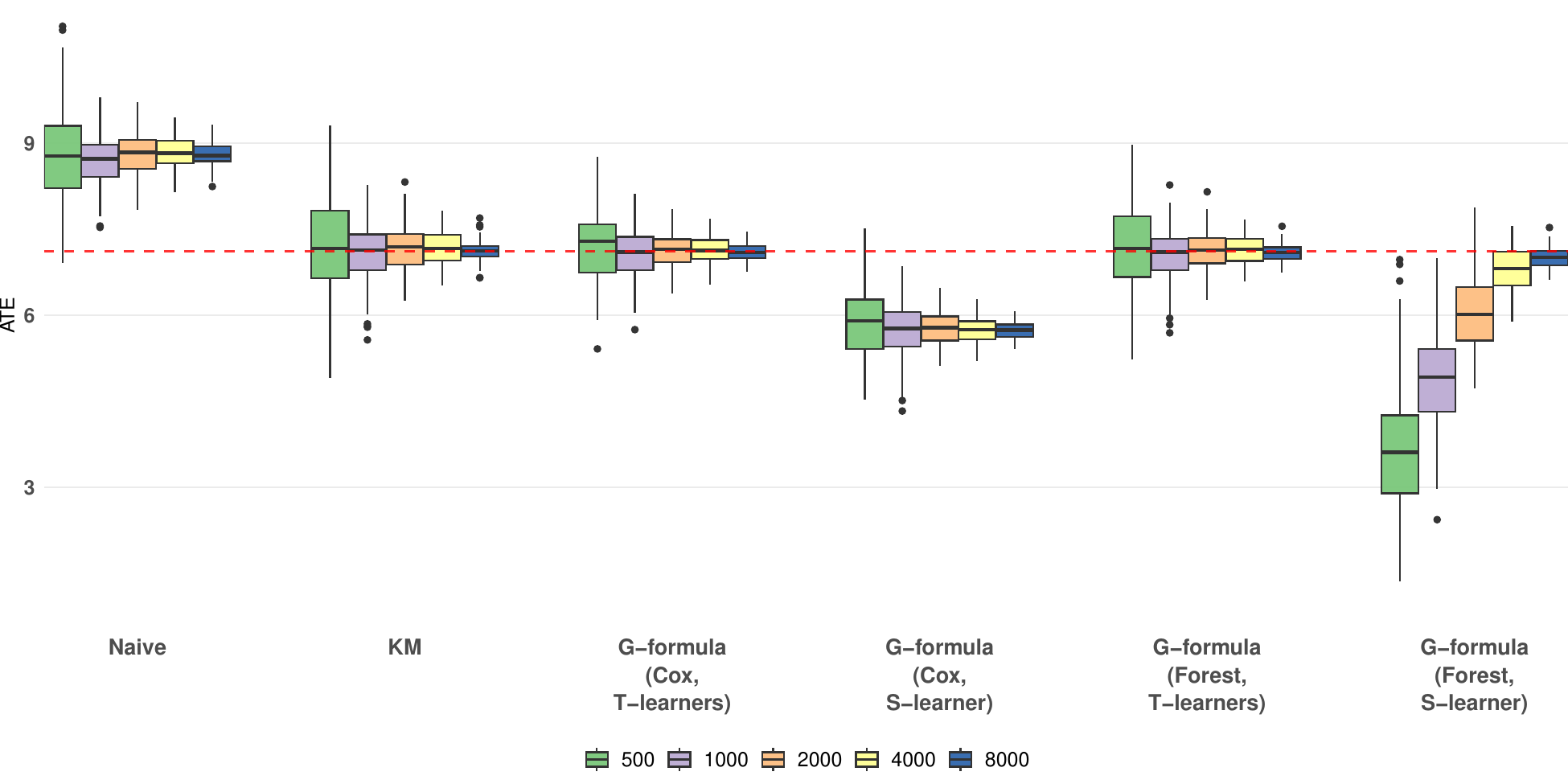}
}
\caption{\label{fig-rct1}{ATEs estimation in RCT setting  with independent censoring.}}

\end{figure}%
In this setting, and in accordance with the theory, the simplest
estimator (unadjusted KM) performs just as well as the others, and
presents, as expected, an extremely small bias.

The naive estimator is biased, as expected, and the bias in the G-formula S-learner arises because the treatment effect is additive \(T^ {(1)} = T^ {(0)} + 10\) and violates the assumption that \(T\) would follow a Cox model in the variables \((X,A)\). However, \(T|A=a\) is a Cox-model for \(a \in \{0,1\}\), which explains the remarkable performance of G-formula (Cox/T-learners) achieves the lowest finite-sample variance. \textcolor{black}{Thus, this discrepancy reflects misspecification of the working Cox model rather than a limitation of the S-learner framework itself.} Some forest-based versions also yield nearly unbiased results, although their performance can vary: for instance, the T-learner remains well-behaved, while the S-learner appears more sensitive and may require larger sample sizes to achieve convergence.
\textcolor{black}{For completeness, we have also evaluated the other estimators (e.g. IPTW KM, IPCW, IPCW KM, BJ, IPTW-IPCW, IPTW-IPCW KM, IPTW-BJ, AIPTW-AIPCW, causal survival forest) in Figure~\ref{fig-rct1_full} (Appendix~\ref{appendix-c-figures}). As expected, they yield no convergence gains, exhibit larger finite-sample variance, and rely on unnecessary model specification.}

\begin{figure}[h]
\centering{
\hspace{-0.5cm}
\includegraphics[width=0.9\textwidth]{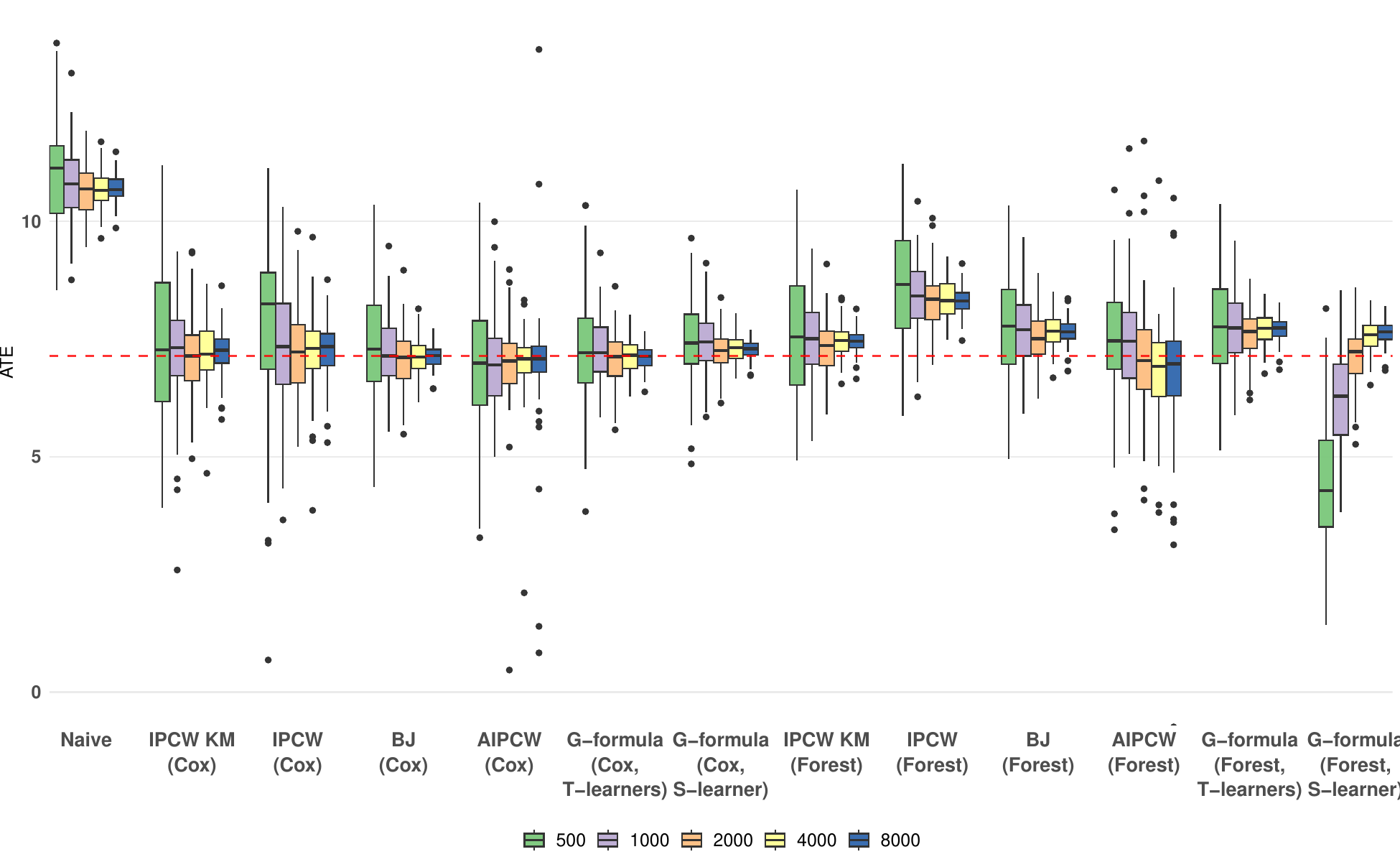}
}
\caption{\label{fig-rct2}{ATEs estimation in RCT setting  with conditionally independent censoring.}}
\end{figure}%

Figure~\ref{fig-rct2} shows the results for the RCT simulation with conditionally independent censoring (Scenario 2), focusing on a subset of estimators.
The naive estimator is biased, as expected. \textcolor{black}{The IPCW (Cox) and is slightly biased up to 2,000 observations. In contrast, both IPCW KM (Cox) and Buckley–James (Cox) and AIPCW (Cox) are unbiased even with only 500 observations, with  a smaller variance for BJ.} G-formula (Cox/T-learners)  performs well, even with small sample size. As in Scenario 1, G-formula (Cox/ S-learner) remains biased. 
The forest-based estimators are biased, \textcolor{black}{except for AIPCW, which is the most stable within this class. Although it still shows a small bias, it is markedly smaller than that of the other methods and appears to converge.}
Notably, all estimators exhibit higher variability compared to Scenario 1 likely \textcolor{black}{due to the higher censoring rate in this scenario and the estimation of censoring model.}

\textcolor{black}{For completeness, we have also evaluated the other estimators (e.g. IPTW KM, IPTW-IPCW, IPTW-IPCW KM, IPTW-BJ, AIPTW-AIPCW, causal survival forest) in Figure~\ref{fig-rct2_full} (Appendix~\ref{appendix-c-figures}). Introducing a superfluous treatment model via the propensity score inflates finite-sample variance, with a marked effect at 500 observations.}

\subsection{Observational data}\label{sec-simulation-Obs}

We simulate an observational study under both independent and conditionally independent censoring (as in Scenarios~1–2). The data consist of an i.i.d. \(n\)-sample \((X_i,A_i,C_i,T_i^{(0)},T_i^{(1)})_{i\in[n]}\) generated as in Section~\ref{sec-simulation-RCT}, except that treatment allocation now follows a logistic propensity model, $\operatorname{logit}(e(X))=\beta_A^\top X,$
so that \(e(X)\) is no longer constant. The survival data-generating mechanism is unchanged; hence the target parameter \(\theta_{\mathrm{RMST}}\) coincides with that of the RCT simulations, namely \(7.1\) (see Section~\ref{sec-simulation-RCT}).

Figure~\ref{fig-obs1} below shows the distribution of the estimators of
\(\theta_{\mathrm{RMST}}\) for the observational study with independent
censoring.
\begin{figure}[h]
\centering{
\hspace*{-0.05\textwidth}
\includegraphics[width=0.9\textwidth]{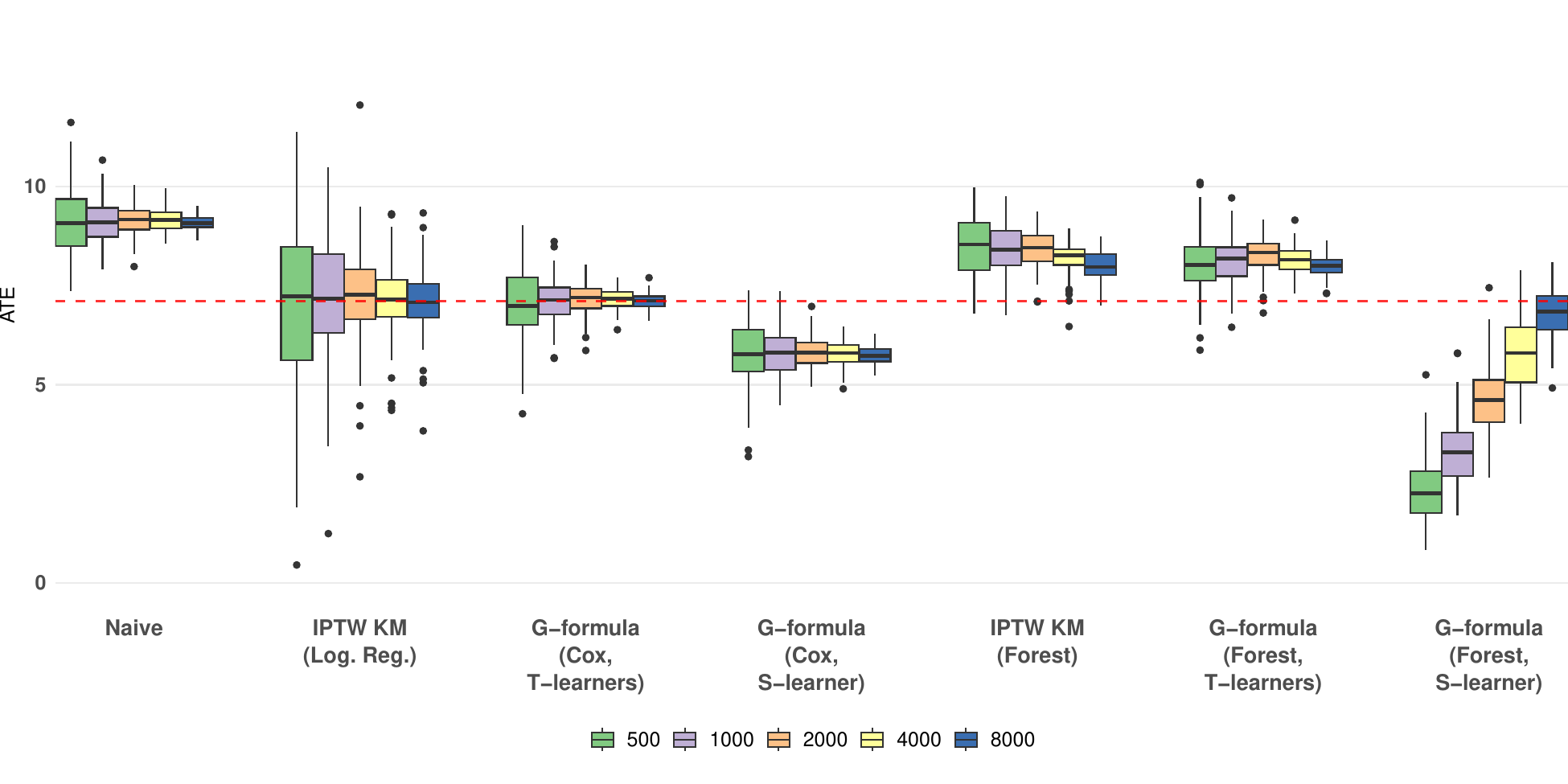}
}
\caption{\label{fig-obs1}{ATEs estimation in observational setting  with independent censoring.}}
\end{figure}%
As expected, the naive estimator is biased while IPTW KM (Log.Reg.) is unbiased. The top-performing estimator in this scenario is the G-formula (Cox/T-learners), which remains unbiased even with as few as 500 observations and achieves the lowest variance. All estimators that use forests to estimate nuisance parameters are biased across sample sizes from 500 to 8000. This setting thus
highlights that one should either have an a priori knowledge of the
specification of the models or large sample size.

Figure~\ref{fig-obs2} below shows the distribution of the
\(\theta_{\mathrm{RMST}}\) estimates for the observational study with
conditionally independent censoring.
\begin{figure}[h]
\centering{
\hspace*{-0.05\textwidth}
\includegraphics[width=0.95\textwidth]{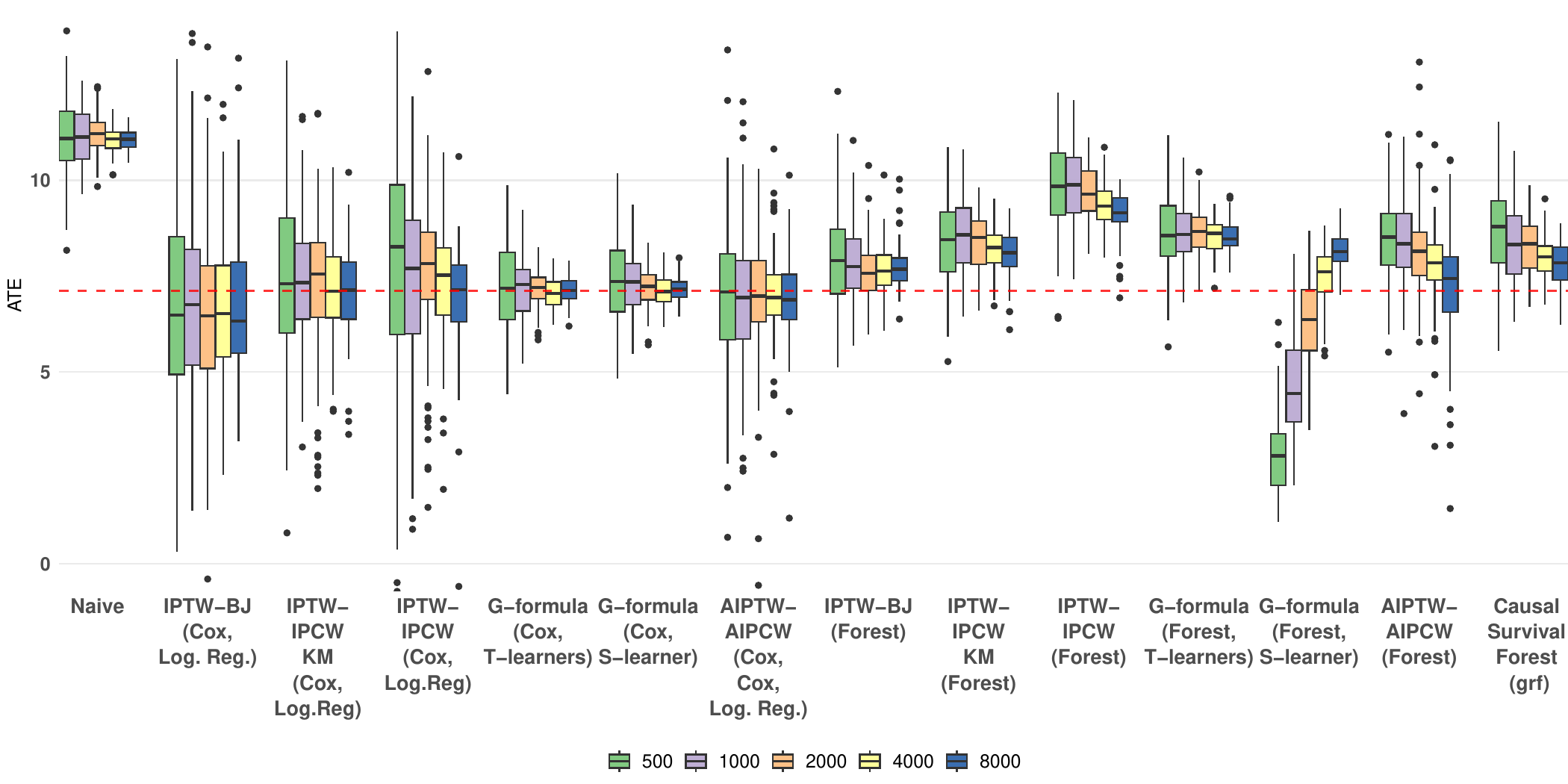}
}
\caption{\label{fig-obs2}{ATEs estimation in observational setting  with conditionally independent censoring.}}
\end{figure}%
The top-performing estimators in this scenario
are G-formula (Cox/T-learners) and AIPCW-AIPTW (Cox \& Cox \&
Log.Reg.), which are unbiased even with 500 observations. The G-formula (Cox/T-learners) has the lowest variance.
Surprisingly, the G-formula (Cox/S-learner) performs quite competitively, showing only a slight bias despite the violation of the proportional hazards assumption. \textcolor{black}{IPTW-IPCW KM (Cox \& Log.Reg.) and IPTW-IPCW (Cox \& Log.Reg.) exhibit different finite-sample behaviors. The Kaplan–Meier version is already unbiased with 500 observations, whereas the moment-based estimator requires around 8,000 observations to achieve unbiasedness. This pattern mirrors the results observed for IPCW and IPCW KM in the randomized trial scenario with conditionally independent censoring.}
All estimators that use forests to estimate nuisance parameters are biased across sample sizes from 500 to 8000. Although Causal Survival Forest \citep{Cui2023} and AIPTW-AIPCW (Forest) are expected to converge, they remain extremely demanding in terms of sample size. 

\textcolor{black}{For completeness, we have also evaluated the other estimators (e.g. KM, IPCW KM, IPCW, BJ, IPTW KM and package versions) in Figure~\ref{fig-obs2_full} (Appendix~\ref{appendix-c-figures})}.

\subsection{Misspecification of nuisance components}\label{sec-simulation-mis}

We generate an observational study with covariate interactions and conditionally independent censoring. The objective is to assess the impact of misspecifying nuisance components; specifically, we will use models that omit interactions to estimate these components. This
approach enables us to evaluate the robustness properties of various estimators. In addition, in this setting forest-based methods are expected to behave better.

We generate an i.i.d. sample \((X_i,A_i,C_i,T_i^{(0)},T_i^{(1)})_{i=1}^n\) where
\(X\sim\mathcal N(\mu,\Sigma)\), the potential event times and censoring follow Cox models with interactions, and treatment follows a logistic propensity with interactions. Specifically, the control potential time has hazard $\lambda^{(0)}(t\mid X)=\exp\{\beta_0^\top Z(X)\},$ the treated potential time is obtained via a constant shift,
$T^{(1)}=T^{(0)}+0.3,$ censoring follows $\lambda_C(t\mid X)=\exp\{\beta_C^\top Z(X)\},$ and the treatment allocation satisfies $\operatorname{logit}(e(X))=\beta_A^\top Z(X).$
Under correct specification, the nuisance models use \((X,Z(X))\) as inputs; under misspecification, only \(X\) and a strict subset of Z(X) (the main-effect terms) are supplied. The true RMST difference at time \(1\) is estimated to be \(0.21\).

We first confirm consistency by estimating \(\theta_{\mathrm{RMST}}\) at 10,000 observations with correctly specified nuisance models and assess robustness under (i) single-model misspecification (Appendix, Figure~\ref{fig-mis2}) and (ii) paired or more misspecifications—(outcome, censoring), (outcome, treatment), (treatment, censoring) and (outcome, treatment, censoring) (Appendix, Figure~\ref{fig-mis3}). \textcolor{black}{Given that, in the parametric setting, IPTW-IPCW KM and IPTW-IPCW behaved similarly asymptotically in the previous simulations, we report results for IPTW-IPCW KM only.}

\textcolor{black}{As expected, when exactly one nuisance component is misspecified (censoring, treatment, or outcome; see Appendix, Figure~\ref{fig-mis2}), AIPTW–AIPCW (Cox \& Cox \& Log.Reg) remains unbiased. In contrast, the other estimators are biased when any of their nuisance models is misspecified; for example, if the treatment model is misspecified, IPTW-IPCW KM is biased. This finding was previously reported in \cite{Ozenne20_AIPTW_AIPCW}.}

When two nuisance components are misspecified (see Appendix, figure~\ref{fig-mis3}), AIPTW-AIPCW (Cox \& Cox \& Log.Reg.) seems to converge in case where the treatment and censoring models are misspecificed.
\textcolor{black}{Furthemore, when all nuisance models are misspecified, all estimators exhibit bias.}

\section{Conclusion}\label{sec-conclusion}

\textbf{Summary of objectives and findings} \textcolor{black}{Restricted mean survival time (RMST)-based estimators provide a valuable alternative to hazard ratios, offering interpretable and robust effect estimates that are particularly well-suited for complex, real-world data. In this work, we provided a comprehensive overview of causal survival estimators targeting RMST, emphasizing both their theoretical foundations and practical performance. As well as implementation (available on a public \href{https://github.com/Sanofi-Public/causal_survival_analysis}{GitHub} repository), we reviewed a wide range of methods including regression-based models, weighting approaches, or hybrid estimators-and assessed their behavior in both their parametric and non-parametric versions through extensive simulations. These scenarios covered diverse settings, including randomized trials and observational data, with varying censoring mechanisms and model specifications. Particular attention was given to robustness, sensitivity to misspecification, and estimator variability in finite samples.}

\textbf{Practical recommendations for RCTs} \textcolor{black}{Whether subject to independent or dependent censoring, we advise against overcomplicating estimation strategies. Classical estimators like the Kaplan-Meier estimator (with weights for conditionally independent censoring if necessary) remain highly effective in these settings without requiring adjustment on covariates. We also recommend using the parametric G-formula. Although randomization can make certain covariates seem unnecessary, theoretical results \citep{Karrison2018} and our simulations show that adjusting for prognostic covariates reduces the asymptotic variance of the RMST-difference estimator. 
\textcolor{black}{The G-formula may be implemented as either an S- or a T-learner. While the S-learner appeared less competitive in some of our simulations, this was primarily attributable to misspecification of the Cox working model rather than to the learning framework itself.} We also advocate for the use of the Buckley-James (BJ) transformation in this context, which demonstrated excellent performance in our simulations, particularly in reducing bias under independent or dependent censoring without introducing excessive variability.} \textcolor{black}{The augmented IPCW (AIPCW) estimator likewise performed very well under conditionally independent censoring and remains particularly appealing due to its double-robust property. Moreover, when nuisance models are estimated nonparametrically, it is the only censoring transformation among those considered that retains formal asymptotic guarantees. However, it is substantially more demanding to implement in practice. By contrast, BJ is simpler and provides a practical and efficient alternative, although it remains sensitive to misspecification.}

\textbf{Practical recommendations for Observational studies} \textcolor{black}{Non-robust weighted estimators such as IPTW-IPCW and IPTW-BJ were consistently marked by high variability, particularly when applied to dependent censoring. Parametric G-formula estimators yielded the best performance under correct model specification (as in \cite{denz2023}), showing both low finite sample bias and variance. However, their sensitivity to model misspecification remains a limitation. Under dependent censoring, we therefore recommend more robust estimators such as AIPTW-AIPCW, which maintain good performance despite their complexity.} 
\textcolor{black}{Parametric estimators perform best when correctly specified; yet with interactions or nonlinearities, nonparametric methods can dominate, as they both avoid functional assumptions and automatically learn complex interactions that parametric models would require explicit specification for.} That said, parametric versions may still be preferable in small to moderate sample sizes, as they exhibited lower bias and variance in our simulations, especially for samples under 4,000 observations. In contrast, when sample sizes are sufficiently large, non-parametric estimators become viable alternatives. \textcolor{black}{In such cases, we specifically recommend the AIPTW-AIPCW estimator, since the nonparametric G-formula lacks formal asymptotic guarantees in this setting and remains inherently prone to plug-in bias, which may persist even in large samples \citep{chernozhukov2018}.}

\textbf{Limitations} \textcolor{black}{It is important to highlight that our simulation design may disadvantage nonparametric methods, as survival and censoring times were generated from exponential distributions and treatment assignments via logistic regression, naturally favoring parametric models.}

\textbf{Perspectives}
\textcolor{black}{Future work should address variable selection in causal survival analysis. In causal inference, selection serves for first identification to adjust for confounders while avoiding instrumental variables (IVs) that inflate variance \citep{lunceford2004} and second to efficiency—include precision variables predictive of outcome \citep{lunceford2004,rotnitzky2019}. With survival data, censoring introduces a third class: covariates tied only to censoring; whether to include them (e.g., in IPCW) is unclear, as they may reduce bias in RMST differences yet raise variance. Systematic evaluation of confounders, precision variables, IVs, and censoring-only covariates on RMST bias–variance across modeling strategies is warranted.}
\textcolor{black}{Other estimators could be benchmarked such as pseudo-observation approaches \citep{andersen2017, wang2018, vanhage2024} and targeted maximum likelihood estimation (TMLE) \citep{Stitelman2012, Talbot2025}.}

\subsubsection*{Acknowledgment}
The research was supported by Sanofi and INRIA (Institut National de Recherche en Informatique et en Automatique) through the CIFRE program (Convention Industrielle de Formation par la Recherche). Charlotte Voinot and Bernard Sebastien are Sanofi employees and may hold shares and/or stock options in the company.  
The work of CB was supported by the Deutsche Foschungsgemeinschaft (German Research Foundation) on the French-German PRCI ANR ASCAI CA 1488/4-1 “Aktive und Batch-Segmentierung, Clustering und Seriation: Grundlagen der KI”.
Imke Mayer and Julie Josse have nothing to disclose.

\vspace*{1pc}
\subsubsection*{Conflict of Interest}
The authors declare no conflicts of interest.

\bibliographystyle{apalike}  
\bibliography{references}

\newpage

\section*{Appendix}
\phantomsection
\label{appendix}

\subsection*[Appendix A: Proofs]{A.1. Appendix A: Proofs}\label{appendix-a-proofs}

\paragraph{Convergence notions and consistency.} Let \((\hat\theta_n)_{n\in\mathbb{N}}\) be a sequence of estimators for a parameter \(\theta\). We recall three standard notions of convergence, commonly used to define statistical consistency:

\begin{enumerate}
  \item \textbf{Weak consistency} means that \(\hat\theta_n\) \textit{converges in probability} to \(\theta\), i.e.,
  \[
  \hat\theta_n \overset{P}{\to} \theta \quad \text{if and only if} \quad \forall \varepsilon > 0, \; \mathbb{P}(|\hat\theta_n - \theta| > \varepsilon) \to 0 \quad \text{as } n \to \infty.
  \]
  This implies that the estimator gets arbitrarily close to the true value with high probability as the sample size increases.

  \item \textbf{Strong consistency} means that \(\hat\theta_n\) \textit{converges almost surely} to \(\theta\), i.e.,
  \[
  \hat\theta_n \overset{\text{a.s.}}{\to} \theta \quad \text{if and only if} \quad \mathbb{P}\left( \lim_{n \to \infty} \hat\theta_n = \theta \right) = 1.
  \]
  This is a stronger form of convergence, requiring the sequence to converge pointwise on almost every sample path.

  \item \textbf{Convergence in distribution} (or \textbf{weak convergence} in the sense of probability theory) refers to the convergence of the distribution of \(\hat\theta_n\) to that of a limiting random variable \(Z\), not necessarily equal to \(\theta\). It is denoted:
  \[
  \hat\theta_n \overset{d}{\to} Z,
  \]
  and defined by the condition that for every bounded continuous function \(f\),
  \[
  \mathbb{E}[f(\hat\theta_n)] \to \mathbb{E}[f(Z)].
  \]
  This notion is typically used in the context of asymptotic distributional results, such as central limit theorems.
\end{enumerate}

\subsubsection*{A.1.1. Proofs of Section~\ref{sec-theoryRCT_indc}}\label{sec-proof21}

\textcolor{black}{Note that not all proofs are novel work. Consistency results on many survival estimators have been reported multiple times as explained in the main paper. For clarity we still recall them. We indicate when the proofs are not novel or when similar proofs exist elsewhere. When we indicate nothing, this means that we have not found those results in other published work.}

\begin{proposition}[Consistency of the Kaplan-Meier estimator]\protect\hypertarget{prp-km}{}\label{prp-km} (Adapted from \cite{gill1983}, \cite{zhou1988})

Under Assumptions
\ref{ass:sutva} (\nameref{ass:sutva}), \ref{ass:rta} (\nameref{ass:rta}),  \ref{ass:trialpositivity} (\nameref{ass:trialpositivity}), \ref{ass:independentcensoring} (\nameref{ass:independentcensoring})
and \ref{ass:poscen}  (\nameref{ass:poscen}), and for all \(t \in [0,\tau]\), the estimator
\(\widehat S_{\mathrm{KM}}(t|A=a)\) of \(S^{(a)}(t)\) is strongly
consistent and admits the following bounds for its bias: \[
0 \leqslant S^{(a)}(t) - \mathbb{E}[\widehat S_\mathrm{KM}(t|A=a)] \leqslant O(\mathbb{P}(N_k(a) = 0)),
\] where \(k\) is the greatest time \(t_k\) such that
\(t \geqslant t_k\).

\end{proposition}

The bound we give, although slightly looser, still exhibits the same
asymptotic regime. In particular, as soon as \(S^{(a)}(t) > 0\) (and
Assumption \ref{ass:poscen} holds), then the bias decays exponentially
fast towards \(0\).

\begin{proof}
\emph{(Proposition~\ref{prp-km}).} Consistency is a trivial consequence of the law of large number. To show that
\(\widehat S_{\mathrm{KM}}\) is unbiased, let us introduce
\(\mathcal{F}_k\) be the filtration generated by the set of variables \[
\{A_i, \mathbb{I}\{\widetilde T_i = t_j\}, \mathbb{I}\{\widetilde T_i = t_j, \Delta_i=1\}~|~j \in [k], i \in [n]\}.
\] which corresponds to the known information up to time \(t_k\), so
that \(D_k(a)\) is \(\mathcal{F}_k\)-measurable but \(N_k(a)\) is
\(\mathcal{F}_{k-1}\)-measurable. One can write that, for
\(k\geqslant 2\) \begin{align*}
\mathbb{E}[\mathbb{I}\{\widetilde T_i = t_k, \Delta_i = 1, A_i=a\}~|~\mathcal{F}_{k-1}] &= \mathbb{E}[\mathbb{I}\{\widetilde T_i = t_k, \Delta_i = 1, A_i=a\}~|~\mathbb{I}\{\widetilde T_i \geqslant t_k\},A_i]  \\
&= \mathbb{I}\{A_i=a\} \mathbb{E}[\mathbb{I}\{T_i = t_k, C_i \geqslant t_k\}~|~\mathbb{I}\{T_i \geqslant t_k, C_i \geqslant t_k\}, A_i] \\
&=  \mathbb{I}\{A_i=a\} \mathbb{I}\{C_i \geqslant t_k\} \mathbb{E}[\mathbb{I}\{T_i = t_k\} ~|~ \mathbb{I}\{T_i \geqslant t_k\}, A_i] \\
&= \mathbb{I}\{\widetilde T_i \geqslant t_k, A_i=a\}\left(1- \frac{S^{(a)}(t_{k})}{S^{(a)}(t_{k-1})}\right),
\end{align*} where we used that \(T_i(a)\) is independent from \(A_i\) by
Assumption~\ref{ass:rta} (\nameref{ass:rta}). We then easily derive from this that \[
\mathbb{E}\left[\left(1-\frac{D_k(a)}{N_k(a)} \right) \mathbb{I}\{N_k(a) >0\}\middle |\mathcal{F}_{k-1}\right] = \frac{S^{(a)}(t_k)}{S^{(a)}(t_{k-1})} \mathbb{I}\{N_k(a) >0\},
\] and then that \[
\mathbb{E}\left[\widehat S_{\mathrm{KM}}(t_k|A=a)\middle |\mathcal{F}_{k-1}\right] =  \frac{S^{(a)}(t_k)}{S^{(a)}(t_{k-1})} \widehat S_{\mathrm{KM}}(t_{k-1}|A=a) + O(\mathbb{I}\{N_k(a) =0\}),
\]

By induction, we easily find that \[
\mathbb{E}[\widehat S_{\mathrm{KM}}(t|A=a)] = \prod_{t_j \leqslant t} \frac{S^{(a)}(t_j)}{S^{(a)}(t_{j-1})} + O\left(\sum_{t_j \leqslant t}\mathbb{P}(N_j(a) =0)\right)= S^{(a)}(t) + O(\mathbb{P}(N_k(a) =0))
\] where \(t_k\) is the greatest time such that \(t_k \leqslant t\).
\end{proof}

\begin{proposition}[Variance of the Kaplan-Meier estimator]\protect\hypertarget{prp-varkm}{}\label{prp-varkm}(Adapted from \cite{Kaplan1958_1})

Under Assumptions
\ref{ass:sutva} (\nameref{ass:sutva}), \ref{ass:rta} (\nameref{ass:rta}), \ref{ass:trialpositivity} (\nameref{ass:trialpositivity}) , \ref{ass:independentcensoring} (\nameref{ass:independentcensoring}) and \ref{ass:poscen} (\nameref{ass:poscen}), and for all \(t \in [0,\tau]\),
\(\widehat S_{\mathrm{KM}}(t|A=a)\) is asymptotically normal and
\(\sqrt{n}\left(\widehat S_{\mathrm{KM}}(t|A=a) - S^{(a)}(t)\right)\)
converges in distribution towards a centered Gaussian of variance \[
V_{\mathrm{KM}}(t|A=a) := S^{(a)}(t)^2 \sum_{t_k \leqslant t} \frac{1-s_k(a)}{s_k(a) r_k(a)},
\] where \(s_k(a) = S^{(a)}(t_k)/S^{(a)}(t_{k-1})\) and
\(r_k(a) = \mathbb{P}(\widetilde T \geqslant t_k, A=a)\).

\end{proposition}

The proof of Proposition~\ref{prp-varkm} is derived below. Because \(D_k(a)/N_k(a)\) is a natural
estimator of \(1-s_k(a)\) and, \(\frac{1}{n} N_k(a)\) a natural
estimator for \(r_k(a)\), the asymptotic variance of the Kaplan-Meier
estimator can be estimated with the so-called Greenwood formula, as
already derived heuristically in \cite{Kaplan1958_1}:

\begin{equation}\phantomsection\label{eq-varkm}{
\widehat{\mathrm{Var}} \left(\widehat{S}_{\mathrm{KM}}(t|A=a)\right) := \widehat{S}_{\mathrm{KM}}(t|A=a)^2 \sum_{t_k \leqslant t} \frac{D_k(a)}{N_k(a)(N_k(a)-D_k(a))}.
}\end{equation}

\begin{proof}
\emph{(Proposition~\ref{prp-varkm}).} The asymptotic normality is a mere
consequence of the joint asymptotic normality of
\((N_k(a),D_k(a))_{t_k \leqslant t}\) with an application of the
\(\delta\)-method. To access the asymptotic variance, notice that, using
a similar reasonning as in the previous proof: \begin{align*}
\mathbb{E}[(1-D_k(a)/N_k(a))^2|\mathcal{F}_{k-1}] &= \mathbb{E}[1-D_k(a)/N_k(a)|\mathcal{F}_{k-1}(a)]^2+\frac{1}{N_k(a)^2}\mathrm{Var}(D_k(a)|\mathcal{F}_{k-1}) \\
&= s_k^2(a)+ \frac{s_k(a)(1-s_k(a))}{N_k(a)}\mathbb{I}\{N_k(a) > 0\} + O(\mathbb{I}\{N_k(a) = 0\}). 
\end{align*} Now we know that
\(N_k(a) = n r_k(a) + \sqrt{n} O_{\mathbb{P}}(1)\), with the
\(O_{\mathbb{P}}(1)\) having uniformly bounded moments. So that we
deduce that \[
\begin{aligned}
\mathbb{E}[(1-D_k(a)/N_k(a))^2|\mathcal{F}_{k-1}] &= s_k^2(a)+ \frac{s_k(a)(1-s_k(a))}{n r_k(a)} +  \frac{1}{n^{3/2}} O_{\mathbb{P}}(1), 
\end{aligned}
\] where \(O_{\mathbb{P}}(1)\) has again bounded moments. Using this
identity, we find that \[
\begin{aligned}
n \mathrm{Var}\widehat S_{\mathrm{KM}} (t|A=a) &= n \left(\mathbb{E}S_{\mathrm{KM}} (t|A=a)^2 -  S^{(a)} (t)^2 \right) \\
&= n S^{(a)}(t)^2 \left(\mathbb{E}\left[\prod_{t_k \leqslant t}  \left(1+\frac1n\frac{1-s_k(a)}{s_k(a) r_k(a)} +   \frac{1}{n^{3/2}} O_{\mathbb{P}}(1) \right)\right]-1\right).
\end{aligned}
\] Expending the product and using that the \(O_{\mathbb{P}}(1)\)'s have
bounded moments, we finally deduce that \[
\begin{aligned}
\mathbb{E}\left[\prod_{t_k \leqslant t}  \left(1+\frac1n\frac{1-s_k(a)}{s_k(a) r_k(a)} +   \frac{1}{n^{3/2}} O_{\mathbb{P}}(1) \right)\right]-1 =   \frac1n \sum_{t_k \leqslant t}\frac{1-s_k(a)}{s_k(a) r_k(a)} +   \frac{1}{n^{3/2}} O(1),
\end{aligned}
\]

\[
n\mathrm{Var}\widehat S_{\mathrm{KM}} (t|A=a) = V_{\mathrm{KM}}(t|A=a) + O(n^{-1/2}),  
\] which is what we wanted to show.
\end{proof}

\subsubsection*{A.1.2. Proofs of Section~\ref{sec-condcens}}\label{sec-proof22}

\begin{proposition}[IPCW is a Censoring Unbiased Transformation]\protect\hypertarget{prp-ipcw}{}\label{prp-ipcw}

Under Assumptions \ref{ass:sutva} (\nameref{ass:sutva}), \ref{ass:rta} (\nameref{ass:rta}) , \ref{ass:trialpositivity} (\nameref{ass:trialpositivity}), \ref{ass:condindepcensoring} (\nameref{ass:condindepcensoring})
and \ref{ass:positivitycensoring} (\nameref{ass:positivitycensoring}), the IPCW transform \ref{eq-defipcw} is
a censoring-unbiased transformation in the sense of
Equation~\ref{eq-cut}.

\end{proposition}

\begin{proof}
\emph{(Proposition~\ref{prp-ipcw}).} Assumption
\ref{ass:positivitycensoring} allows the transformation to be
well-defined. Furthermore, it holds \begin{align*}
E[T^*_{\mathrm{IPCW}}|A=a,X] 
&= E\left[\frac{ \Delta^\tau \times\widetilde T \wedge \tau}{G(\widetilde T \wedge \tau | A,X)} \middle | A = a,X\right]  \\
&= E\left[\frac{ \Delta^\tau \times T^ {(a)} \wedge \tau}{G(T^ {(a)} \wedge \tau | A,X)} \middle | A = a,X\right]\\
&= E\left[ E\left[\frac{ \mathbb{I}\{T^ {(a)}\wedge \tau \leqslant C \} \times T^ {(a)} \wedge \tau}{G(T^ {(a)} \wedge \tau | A,X)}\middle| A, X,T^ {(a)} \right] \middle | A = a,X\right] \\
&= E\left[T^ {(a)} \wedge \tau|A=a, X\right] \\
&= E\left[T^ {(a)} \wedge \tau | X \right].
\end{align*} We used in the second equality that on the event
\(\{\Delta^\tau=1, A=a\}\), it holds
\(\widetilde T \wedge \tau =  T \wedge \tau = T^ {(a)} \wedge \tau\). We
used in the fourth equality that
\(G(T^ {(a)} \wedge \tau | A,X) = E[\mathbb{I}\{T^ {(a)} \wedge \tau \leqslant C\}|X,T^ {(a)},A]\)
thanks to Assumption \ref{ass:condindepcensoring}, and in the last one
that \(A\) is independent of \(X\) and \(T^ {(a)}\) thanks to Assumption
\ref{ass:rta}.
\end{proof}

\begin{proposition}[Consistency of $S^*_{\mathrm{IPCW}}(t|A=a)$ with known $G$]\protect\hypertarget{prp-ipcwkm}{}\label{prp-ipcwkm}

Under Assumptions~\ref{ass:sutva} (\nameref{ass:sutva}), \ref{ass:rta} (\nameref{ass:rta}), \ref{ass:trialpositivity} (\nameref{ass:trialpositivity}), \ref{ass:condindepcensoring} (\nameref{ass:condindepcensoring})
and \ref{ass:positivitycensoring} (\nameref{ass:positivitycensoring}), and for all \(t \in [0,\tau]\), the
oracle estimator \(S^*_{\mathrm{IPCW}}(t|A=a)\) defined as in
Definition~\ref{def-ipcwkm} with \(\widehat G = G\) is a stronlgy
consistent and asymptotically normal estimator of \(S^{(a)}(t)\) .

\end{proposition}

\begin{proof}
\emph{(Proposition~\ref{prp-ipcwkm}).} Similarly to the computations
done in the proof of Proposition~\ref{prp-ipcw}, it is easy to show that, for $t_k \leq \tau$,
\begin{align*}
\mathbb{E}\left[\frac{\Delta^\tau}{G( t_k | X,A)} \mathbb{I}(\widetilde T = t_k, A=a)\right] &= \mathbb{E}\left[\frac{\mathbb{I}(t_k \leq C)}{G( t_k | X,A)} \mathbb{I}(T^{(a)} = t_k, A=a)\right], \\
&= \mathbb{P}(A=a)\mathbb{P}(T^{(a)}=t_k),
\end{align*}
and likewise that 
 \begin{align*}
     \mathbb{E}\left[\frac{1}{G(t_k | X,A)} \mathbb{I}(\widetilde T \geqslant t_k, A=a)\right] &=\mathbb{E}\left[\frac{1}{G(t_k | X,A)} \mathbb{I}(T^{(a)} \wedge C \geqslant t_k, A=a)\right] \\
     &=\mathbb{E}\left[\frac{\mathbb{I}(T^{(a)} \geq t_k)\mathbb{I}(A=a)}{G(t_k|X,A)}\mathbb{E}\left[ \mathbb{I}(C \geqslant t_k)|X,A\right]\right] \\
     &= \mathbb{P}(A=a) \mathbb{P}(T^{(a)} \geq t_k),
 \end{align*}
so that \(\widehat S_{\mathrm{IPCW}}(t)\) converges almost surely
towards the product limit \[
\prod_{t_k \leqslant t} \left(1-\frac{\mathbb{P}(T^ {(a)} = t_k)}{\mathbb{P}(T^ {(a)} \geqslant t_k)}\right) = S^{(a)}(t),
\] yielding strong consistency. Asymptotic normality is straightforward.
\end{proof}

\begin{proposition}[BJ is a censoring unbiased transformation]\protect\hypertarget{prp-bj}{}\label{prp-bj}

Under Assumptions\ref{ass:sutva} (\nameref{ass:sutva}), \ref{ass:rta} (\nameref{ass:rta}) , \ref{ass:condindepcensoring} (\nameref{ass:condindepcensoring})
and \ref{ass:positivitycensoring} (\nameref{ass:positivitycensoring}), the BJ transform \ref{eq-defbj} is a
censoring-unbiased transformation in the sense of Equation~\ref{eq-cut}.

\end{proposition}

\begin{proof}
\emph{(Proposition~\ref{prp-bj}).} There holds \[
\begin{aligned}
\mathbb{E}[T_{\mathrm{BJ}}^*|X,A=a] &= \mathbb{E}\left[\Delta^\tau T^ {(a)} \wedge \tau + (1-\Delta^\tau)\frac{\mathbb{E}[T \wedge \tau \times \mathbb{I}\{T \wedge \tau > C\}|C,A,X]}{\mathbb{P}(T > C|C,A,X)} \middle | X,A=a \right] \\
&= \mathbb{E}[\Delta^\tau T^ {(a)} \wedge \tau | X ] +  \underbrace{\mathbb{E}\left[ \mathbb{I}\{T \wedge \tau > C\} \frac{\mathbb{E}[T \wedge \tau \times \mathbb{I}\{T \wedge \tau > C\}|C,A,X]}{\mathbb{E}[\mathbb{I}\{T \wedge \tau \geqslant C\}|C,A,X]}\middle | X,A=a \right]}_{(\star)}.
\end{aligned}
\] Now we easily see that conditionning wrt \(X\) in the second term
yields \begin{align*}
(\star) &= \mathbb{E}\left[ \mathbb{E}[T \wedge \tau \times \mathbb{I}\{T \wedge \tau > C\}|C,A,X] \middle | X,A=a \right] \\
&= \mathbb{E}[(1-\Delta^\tau) T \wedge \tau | X,A=a ]  \\
&= \mathbb{E}[(1-\Delta^\tau)  T^ {(a)} \wedge \tau | X], 
\end{align*} ending the proof.
\end{proof}

\begin{proof}
\emph{(Theorem~\ref{thm-bj}, \nameref{thm-bj}).} We let
\(T^* = \Delta^\tau \phi_1 + (1-\Delta^\tau)\phi_0\) be a transformation
of the form Equation~\ref{eq-defcut}. There holds \[
\mathbb{E}[(T^*-T \wedge \tau)^2] = \mathbb{E}[\Delta^\tau (\phi_1-T \wedge \tau)^2] + \mathbb{E}[(1-\Delta^\tau)(\phi_0-T \wedge \tau)^2].
\] The first term is non negative and is zero for the BJ transformation.
Since \(\phi_0\) is a function of \((\widetilde T, X, A)\) and that
\(\widetilde T = C\) on \(\{\Delta^\tau = 0\}\), the second term can be
rewritten in the following way. We let \(R\) be a generic quantity that
does not depend on \(\phi_0\). \begin{align*}
\mathbb{E}&[(1-\Delta^\tau)(\phi_0-T)^2] = \mathbb{E}\left[\mathbb{I}\{T\wedge \tau > C\} \phi_0^2 - 2 \mathbb{I}\{T\wedge \tau > C\} \phi_0 T \wedge \tau\right]  + R \\
&= \mathbb{E}\left[ \mathbb{P}(T\wedge \tau > C|C,A,X) \phi_0^2 - 2 \mathbb{E}[T\wedge \tau \mathbb{I}\{T\wedge \tau > C\}|C,A,X] \phi_0 \right] + R \\
&= \mathbb{E}\left[\mathbb{P}(T\wedge \tau > C|C,A,X) \left(\phi_0- \frac{\mathbb{E}[T\wedge \tau \mathbb{I}\{T\wedge \tau > C\}|C,A,X]}{\mathbb{P}(T\wedge \tau > C|C,A,X) }\right)^2\right] + R.
\end{align*} Now the first term in the right hand side is always
non-negative, and is zero for the BJ transformation.
\end{proof}

\begin{proposition}[AIPCW is a censoring-unbiased transformation]\protect\hypertarget{prp-tdr}{}\label{prp-tdr}

We let \(F,R\) be two conditional survival functions. Under Assumptions\ref{ass:sutva} (\nameref{ass:sutva}),  \ref{ass:rta} (\nameref{ass:rta}) ,  \ref{ass:trialpositivity} (\nameref{ass:trialpositivity}), \ref{ass:condindepcensoring} (\nameref{ass:condindepcensoring}) and if \(F\) also satisfies Assumption~\ref{ass:positivitycensoring} (\nameref{ass:positivitycensoring}), and if \(F(\cdot|X,A)\) is absolutely
continuous wrt \(G(\cdot|X,A)\), then the transformation
\(T^*_\mathrm{DR} = T^*_\mathrm{DR}(F,R)\) is a censoring-unbiased
transformation in the sense of Equation~\ref{eq-cut} whenever \(F = G\)
or \(R=S\).

\end{proposition}

A careful examination of the proofs from \cite{rubin2007} show that the proof translates straight away
to our discrete setting.

\subsubsection*{A.1.3. Proofs of Section~\ref{sec-obs_indcen}}\label{sec-proof13}

\begin{proposition}[Consistency of $S^*_{\mathrm{IPTW}}(t | A=a)$ with known $e$]\protect\hypertarget{prp-iptwkm}{}\label{prp-iptwkm}

Under Assumptions
\ref{ass:sutva} (\nameref{ass:sutva}),  \ref{ass:unconf} (\nameref{ass:unconf}), \ref{ass:positivitytreat} (\nameref{ass:positivitytreat}), \ref{ass:independentcensoring} (\nameref{ass:independentcensoring}) and \ref{ass:poscen} (\nameref{ass:poscen}) The oracle IPTW Kaplan-Meier estimator
\(S^*_{\mathrm{IPTW}}(t | A=a)\) with $\hat{e}=e$ is a strongly consistent and
asymptotically normal estimator of \(S^{(a)}(t)\).

\end{proposition}

The proof of this result simply relies again on the law of large number
and the \(\delta\)-method.

\begin{proof}
\emph{(Proposition~\ref{prp-iptwkm}).} The fact that it is strongly
consistent and asymptotically normal is again a simple application of
the law of large number and of the \(\delta\)-method. We indeed find
that, for \(t_k \leqslant\tau\) \[
\begin{aligned}
\mathbb{E}\left[\frac{1}{e(X_i)} \mathbbm{1}\{\widetilde T_i = t_k, \Delta_i = 1, A_i=1\}\right] &= \mathbb{E}\left[\frac{A_i}{e(X_i)} \mathbbm{1}\{T_i = t_k, C_i \geqslant t_k\}\right] \\
&=\mathbb{E}\left[ \mathbb{E}\left[ \frac{A_i}{e(X_i)} \mathbbm{1}\{T_i = t_k, C_i \geqslant t_k\} \middle |X_i \right]\right] \\
&= \mathbb{E}\left[ \mathbb{E}\left[\frac{A_i}{e(X_i)}\middle | X_i\right] \mathbb{P}(T_i = t_k |X_i) \mathbb{P}( C_i \geqslant t_k)\right] \\
&= \mathbb{P}(T_i = t_k) \mathbb{P}( C_i \geqslant t_k),
\end{aligned}
\] where we used that \(A\) is independent of \(T\) conditionnaly on
\(X\), and that \(C\) is independent of everything. Likewise, one
would get that \[
\begin{aligned}
\mathbb{E}\left[\frac{1}{e(X_i)} \mathbbm{1}\{\widetilde T_i \geqslant t_k, A_i=1\}\right] =
&= \mathbb{P}(T_i \geqslant t_k) \mathbb{P}( C_i \geqslant t_k).
\end{aligned}
\] Similar computations hold for \(A=0\), ending the proof.
\end{proof}

\subsubsection*{A.1.4. Proofs of Section~\ref{sec-obscondcens}}\label{sec-proof32}

\begin{proposition}[IPTW-IPCW is a censoring unbiased transformation]\protect\hypertarget{prp-iptwipcw}{}\label{prp-iptwipcw}

Under Assumptions
\ref{ass:sutva} (\nameref{ass:sutva}),  \ref{ass:unconf} (\nameref{ass:unconf}), \ref{ass:positivitytreat} (\nameref{ass:positivitytreat}), \ref{ass:condindepcensoring} (\nameref{ass:condindepcensoring})
and \ref{ass:positivitycensoring} (\nameref{ass:positivitycensoring}), the IPTW-IPCW transform
\ref{eq-defipcw} is a censoring-unbiased transformation in the sense of
Equation~\ref{eq-cut}.

\end{proposition}

\begin{proof}
\emph{(Proposition~\ref{prp-iptwipcw}).} On the event
\(\{\Delta^\tau=1, A=1\}\), there holds
\(\widetilde T \wedge \tau = T \wedge \tau = T^ {(1)} \wedge \tau\), whence
we find that, \[
\begin{aligned}
\mathbb{E}[ T^*_{\mathrm{IPCW}}|X,A=1] &= \mathbb{E}\left[\frac{A}{e(X)} \frac{\mathbb{I}\{T^ {(1)} \wedge \tau \leqslant C\}}{G(T^ {(1)} \wedge \tau|X,A)} T^ {(1)} \wedge \tau\middle|X\right] \\
&= \mathbb{E}\left[ \frac{A}{e(X)} \mathbb{E}\left[\frac{\mathbb{I}\{T^ {(1)} \wedge \tau \leqslant C\}}{G(T^ {(1)} \wedge \tau|X,A)} \middle| X,A, T^ {(1)} \right]T^ {(1)} \wedge \tau\middle|X\right] \\
&=\mathbb{E}\left[ \frac{A}{e(X)} T^ {(1)} \wedge \tau\middle|X\right] \\
&=\mathbb{E}\left[T^ {(1)} \wedge \tau\middle|X\right],
\end{aligned}
\] and the same holds on the event \(A=0\).
\end{proof}

\begin{proof}
\emph{(Proposition~\ref{prp-consiptwipcw}, \nameref{prp-consiptwipcw}).} By consistency of
\(\widehat G(\cdot|X,A)\) and \(\widehat e\) and by continuity, it
suffices to look at the asymptotic behavior of the oracle estimator \[
\theta^*_{\mathrm{IPTW-IPCW}} = \frac1n\sum_{i=1}^n  \left(\frac{A_i}{ e(X_i)}-\frac{1-A_i}{1-e(X_i)} \right)\frac{\Delta_i^\tau}{G(\widetilde T_i \wedge \tau | A_i,X_i)} \widetilde T_i \wedge \tau.
\] The later is converging almost towards its mean, which, following
similar computations as in the previous proof, write \[
\begin{aligned}
\mathbb{E}\left[\left(\frac{A}{e(X)}-\frac{1-A}{1-e(X)} \right)\frac{\Delta^\tau}{G(\widetilde T \wedge \tau | A,X)} \widetilde T \wedge \tau\right]  
&= \mathbb{E}\left[\left(\frac{A}{e(X)}-\frac{1-A}{1-e(X)} \right) T \wedge \tau\right] \\
&= \mathbb{E}\left[T^ {(1)} \wedge \tau\right]-\mathbb{E}\left[T^ {(0)} \wedge \tau\right].
\end{aligned}
\]
\end{proof}

\begin{proposition}[Consistency of \(S^*_{\mathrm{IPTW-IPCW}}(t | A=a)\) with known $e$ and $G$]\protect\hypertarget{prp-iptwipcwkm}{}\label{prp-iptwipcwkm}

Under Assumptions
\ref{ass:sutva} (\nameref{ass:sutva}),  \ref{ass:unconf} (\nameref{ass:unconf}), \ref{ass:positivitytreat} (\nameref{ass:positivitytreat}), \ref{ass:condindepcensoring} (\nameref{ass:condindepcensoring})
and \ref{ass:positivitycensoring} (\nameref{ass:positivitycensoring}), and for all \(t \in [0,\tau]\), if the
oracle estimator \(S^*_{\mathrm{IPTW-IPCW}}(t | A=a)\) defined as in
Definition~\ref{def-iptwipcwkm} with
\(\widehat G(\cdot|A,X) = G(\cdot|A,X)\) and \(\widehat e = e\) is a
strongly consistent and asymptotically normal estimator of
\(S^{(a)}(t)\).

\end{proposition}

\begin{proof}
\emph{(Proposition~\ref{prp-iptwipcwkm}).} Asymptotic normality comes
from a mere application of the \(\delta\)-method, while strong
consistency follows from the law of large number and the follozing
computations. Like for the proof of Proposition~\ref{prp-ipcwkm}, one
find, by first conditionning wrt \(X,A,T^ {(a)}\), that, for
\(t_k \leqslant\tau\), \[
\begin{aligned}
\mathbb{E}\left[\left(\frac{A}{e(X)}+\frac{1-A}{1-e(X)} \right)\frac{\Delta^\tau}{G(t_k| A,X)} \mathbb{I}\{\widetilde T = t_k, A=a\}\right] = \mathbb{P}(T^ {(a)}=t_k)
\end{aligned}
\] and likewise that \[
\begin{aligned}
\mathbb{E}\left[\left(\frac{A}{e(X)}+\frac{1-A}{1-e(X)} \right)\frac{1}{G(t_k | A,X)} \mathbb{I}\{\widetilde T \geqslant t_k, A=a\}\right] = \mathbb{P}(T^ {(a)}\geqslant t_k)
\end{aligned}
\] so that indeed \(S^*_{\mathrm{IPTW-IPCW}} (t|A=a)\) converges almost
surely towards \(S^{(a)}(t)\).
\end{proof}

\begin{proposition}[IPTW-BJ is a censoring-unbiased transformation]\protect\hypertarget{prp-iptwbj}{}\label{prp-iptwbj}

Under Assumptions
\ref{ass:sutva} (\nameref{ass:sutva}),  \ref{ass:unconf} (\nameref{ass:unconf}), \ref{ass:positivitytreat} (\nameref{ass:positivitytreat}), \ref{ass:condindepcensoring} (\nameref{ass:condindepcensoring})
and \ref{ass:positivitycensoring} (\nameref{ass:positivitycensoring}), the IPTW-BJ transform \ref{eq-defbj}
is a censoring-unbiased transformation in the sense of
Equation~\ref{eq-cut}.

\end{proposition}

\begin{proof}
\emph{(Proposition~\ref{prp-iptwbj}).} We write \[
\begin{aligned}
\mathbb{E}[ T^*_{\mathrm{IPTW-BJ}}|X,A=1] &= \mathbb{E}\left[\frac{A}{e(X)} \Delta^\tau \times \widetilde T \wedge \tau\middle|X\right] + \mathbb{E}\left[\frac{A}{e(X)} (1-\Delta^\tau) Q_S(\widetilde T \wedge \tau|A,X) \middle| X\right]. 
\end{aligned}
\] On the event \(\{\Delta^\tau=1, A=1\}\), there holds
\(\widetilde T \wedge \tau = T \wedge \tau = T^ {(1)} \wedge \tau\), whence
we find that the first term on the the right hand side is equal to \[
\begin{aligned}
\mathbb{E}\left[\frac{A}{e(X)} \Delta^\tau \times \widetilde T \wedge \tau\middle|X\right] &= \mathbb{E}\left[\frac{A}{e(X)} \Delta^\tau \times T^ {(1)} \wedge \tau\middle|X\right] \\
&= \mathbb{E}\left[\Delta^\tau \times T^ {(1)} \wedge \tau\middle|X\right].
\end{aligned}
\] For the second term in the right hand side, notice that on the event
\(\{\Delta^\tau=0, A=1\}\), there holds
\(\widetilde T = C < T^ {(1)} \wedge \tau\), so that \[
\begin{aligned}
\mathbb{E}&\left[\frac{A}{e(X)} \mathbb{I}\{C < T^ {(1)} \wedge \tau\} \frac{\mathbb{E}[T^ {(1)} \wedge \tau \times \mathbb{I}\{C < T^ {(1)} \wedge \tau\}|X,A,C]}{\mathbb{P}(C < T^ {(1)} \wedge \tau | C,X,A)} \middle| X\right] \\
&= \mathbb{E}\left[\frac{A}{e(X)} \mathbb{E}[T^ {(1)} \wedge \tau \times \mathbb{I}\{C < T^ {(1)} \wedge \tau\}|X,A,C] \middle| X\right] \\
&= \mathbb{E}\left[T^ {(1)} \wedge \tau \times \mathbb{I}\{C < T^ {(1)} \wedge \tau\} \middle| X\right] \\
&= \mathbb{E}\left[(1-\Delta^\tau) T^ {(1)} \wedge \tau  \middle| X\right], 
\end{aligned}
\] and the sane holds on the event \(\{A=0\}\), which ends the proof.
\end{proof}

\begin{proof}
\emph{(Proposition~\ref{prp-consiptwbj}, \nameref{prp-consiptwbj}).} By consistency of
\(\widehat G(\cdot|X,A)\) and \(\widehat e\) and by continuity, it
suffices to look at the asymptotic behavior of the oracle estimator \[
\theta^*_{\mathrm{IPTW-BJ}} = \frac1n\sum_{i=1}^n  \left(\frac{A_i}{ e(X_i)}-\frac{1-A_i}{1-e(X_i)} \right)\left(\Delta_i^\tau \times \widetilde T_i \wedge \tau + (1-\Delta_i^\tau) Q_S(\widetilde T_i \wedge \tau|A_i,X_i)\right).
\] The later is converging almost towards its mean, which, following
similar computations as in the previous proof, is simply equal to the
RMST difference.
\end{proof}

\begin{proposition}[Double/Triple robustness of $\Delta^*_\mathrm{QR}$]\protect\hypertarget{prp-tqr}{}\label{prp-tqr} (adapted from \cite{Ozenne20_AIPTW_AIPCW})

Let \(F,R\) be two conditional survival function, \(p\) be a propensity
score, and \(\nu\) be a conditional response. Then, under the same
assumption on \(F,R\) as in Proposition~\ref{prp-tdr}, and under
Assumptions
\ref{ass:sutva} (\nameref{ass:sutva}),  \ref{ass:unconf} (\nameref{ass:unconf}), \ref{ass:positivitytreat} (\nameref{ass:positivitytreat}), \ref{ass:condindepcensoring} (\nameref{ass:condindepcensoring})
and \ref{ass:positivitycensoring} (\nameref{ass:positivitycensoring}), the transformations
\(\Delta^*_{\mathrm{QR}} = \Delta^*_{\mathrm{QR}}(F,R,p,\nu)\)
satisfies, fo \(a \in {0,1}\), \[
\mathbb{E}[ \Delta^*_\mathrm{QR} |X] = \mathbb{E}[T^ {(1)}\wedge \tau-T^ {(0)}\wedge \tau |X]\quad\text{if}\quad  
\begin{cases} F = G \quad &\text{or}\quad R=S \quad \text{and} \\
p=e \quad &\text{or}\quad \nu=\mu.
\end{cases}
\]

\end{proposition}

\begin{proof}
\emph{(Proposition~\ref{prp-tqr}).} We can write that \[
\Delta^*_{\mathrm{QR}} = \underbrace{\frac{A}{p(X)}(T^*_{\mathrm{DR}}(F,R)-\nu(X,1))+ \nu(X,1)}_{(\mathrm{A})} - \left(\underbrace{\frac{1-A}{1-p(X)}(T^*_{\mathrm{DR}}(F,R)-\nu(X,0))+ \nu(X,0)}_{(\mathrm{B})} \right). 
\] Focusing on term \((\mathrm{A})\), we easily find that \[
\begin{aligned}
\mathbb{E}[(\mathrm{A}) |X] &= \mathbb{E}\left[\frac{A}{p(X)}(T^*_{\mathrm{DR}}(F,R)-\nu(X,1))+ \nu(X,1) \middle|X\right] \\
&= \frac{e(X)}{p(X)}(\mu(X,1)-\nu(X,1)) + \nu(X,1).
\end{aligned}
\] Where we used that \(T^*_{\mathrm{DR}}(F,R)\) is a censoring-unbiased
transform when \(F=G\) or \(R=S\). Now we see that if \(p(X) = e(X)\),
then \[
\mathbb{E}[(\mathrm{A}) |X] = \mu(X,1)-\nu(X,1) + \nu(X,1) = \mu(X,1),
\] and if \(\nu(X,1) = \mu(X,1)\), then \[
\mathbb{E}[(\mathrm{A}) |X] = \frac{e(X)}{p(X)} \times 0 + \mu(X,1) = \mu(X,1).
\] Likewise, we would show that
\(\mathbb{E}[(\mathrm{B}) |X] = \mu(X,0)\) under either alternative,
ending the proof.
\end{proof}

\subsection*{A.2. Appendix B: Implementations}\label{appendix-b-code}

\subsubsection*{A.2.1. Available packages}\label{sec-package}

\textbf{SurvRM2}

The RMST difference with Unadjusted Kaplan-Meier \(\hat \theta_{KM}\)
(Equation~\ref{eq-thetaunadjKM}) can be obtained using the function
\(\texttt{rmst2}\) which requires as input the observed time-to-event,
the status, the arm which corresponds to the treatment and \(\tau\).

\textbf{RISCA}

The \href{https://cran.r-project.org/web/packages/RISCA/index.html}{RISCA} package \citep{Foucher2019} provides several methods for estimating
\(\theta_{\mathrm{RMST}}\). The RMST difference with Unadjusted
Kaplan-Meier \(\hat \theta_{KM}\) (Equation~\ref{eq-thetaunadjKM}) can
be derived using the \(\texttt{survfit}\) function from the  \href{https://cran.r-project.org/web/packages/survival/index.html}{survival}
package \citep{Therneau2001} which estimates Kaplan-Meier survival curves for
treated and control groups, and then the \(\texttt{rmst}\) function
calculates the RMST by integrating these curves, applying the rectangle
method (type=``s''), which is well-suited for step functions.

The IPTW Kaplan-Meier (Equation~\ref{eq-IPTWKM}) can be applied using
the \(\texttt{ipw.survival}\) and \(\texttt{rmst}\) functions. The
ipw.survival function requires user-specified weights (i.e.~propensity
scores). To streamline this process, we define the
\(\texttt{RISCA\_iptw}\) function, which combines these steps and
utilizes the \(\texttt{estimate\_propensity\_score}\) from the
\(\texttt{utilitary.R}\) file.

A single-learner version of the G-formula, as introduced in
Section~\ref{sec-Gformula}, can be implemented using the
\(\texttt{gc.survival}\) function. This function requires as input the
conditional survival function which should be estimated beforehand with
a Cox model via the \(\texttt{coxph}\) function from the
\href{https://cran.r-project.org/web/packages/survival/index.html}{survival} package \citep{Therneau2001}. Specifically, the
single-learner approach applies a single Cox model incorporating both
covariates and treatment, rather than separate models for each treatment
arm. We provide a function \(\texttt{RISCA\_gf}\) that consolidates
these steps.

\textbf{grf}

The \href{https:\%20//cran.r-project.org/web/packages/grf/index.html}{grf} package
\citep{Tibshirani2017} enables estimation
of the difference between RMST using the Causal Survival Forest approach \citep{Cui2023}, which extends the non-parametric causal forest
framework to survival data. The RMST can be estimated with the
\(\texttt{causal\_survival\_forest}\) function, requiring covariates
\(X\), observed event times, event status, treatment assignment, and the
time horizon \(\tau\) as inputs. The
\(\texttt{average\_treatment\_effect}\) function then evaluates the
treatment effect based on predictions from the fitted forest.

\subsubsection*{A.2.2. Implementation of the
estimators}\label{implementation-of-the-estimators}

The complete codebase is available on \href{https://github.com/Sanofi-Public/causal_survival_analysis}{\texttt{\detokenize{github.com/Sanofi-Public/causal_survival_analysis}}}

\textbf{Unadjusted Kaplan-Meier}

Although Kaplan-Meier is implemented in the \href{https://cran.r-project.org/web/packages/survival/index.html}{survival}
package \citep{Therneau2001}, we provide a custom implementation,
\(\texttt{Kaplan\_meier\_handmade}\), for completeness. The difference
in Restricted Mean Survival Time, estimated using Kaplan-Meier as in
Equation~\ref{eq-thetaunadjKM} can then be calculated with the
\(\texttt{RMST\_1}\) function.

As an alternative, one can also use the \(\texttt{survfit}\) function in
the \href{https://cran.r-project.org/web/packages/survival/index.html}{survival} package \citep{Therneau2001} for Kaplan-Meier and specify the
\(\texttt{rmean}\) argument equal to \(\tau\) in the corresponding
summary function:

\textbf{IPCW Kaplan-Meier}

We first provide an \(\texttt{adjusted.KM}\) function which is then used
in the \(\texttt{IPCW\_Kaplan\_meier}\) function to estimate the
RMST difference \(\hat{\theta}_{\mathrm{IPCW}}\) as in
Equation~\ref{eq-thetaIPCWKM}. The survival censoring function
\(G(t|X)\) is computed with the
\(\texttt{estimate\_survival\_function}\) utility function from the
\(\texttt{utilitary.R}\) file.

One could also use the \(\texttt{survfit}\) function in the \href{https://cran.r-project.org/web/packages/survival/index.html}{survival}
package \citep{Therneau2001} by adding IPCW weights for treated and control
group and specify the \(\texttt{rmean}\) argument equal to \(\tau\) in
the corresponding summary function:

This alternative approach for IPCW Kaplan-Meier would also be valid for
IPTW and IPTW-IPCW Kaplan-Meier.

\textbf{Buckley-James-based estimator }

The function \(\texttt{BJ}\) estimates \(\theta_{\mathrm{RMST}}\) by
implementing the Buckley-James estimator as in Equation~\ref{eq-BJnopi}. It
uses two functions available in the \(\texttt{utilitary.R}\) file,
namely \(\texttt{Q\_t\_hat}\) and \(\texttt{Q\_Y}\).

\textbf{IPTW Kaplan-Meier}

The function \(\texttt{IPTW\_Kaplan\_meier}\) implements the IPTW-KM
estimator in Equation~\ref{eq-RMST_IPTWKM}. It uses the
\(\texttt{estimate\_propensity\_score}\) function from the
\(\texttt{utilitary.R}\).

\textbf{G-formula}

We implement two versions of the G-formula:
\(\texttt{g\_formula\_T\_learner}\) and
\(\texttt{g\_formula\_S\_learner}\). In
\(\texttt{g\_formula\_T\_learner}\), separate models estimate survival
curves for treated and control groups, whereas
\(\texttt{g\_formula\_S\_learner}\) uses a single model incorporating
both covariates and treatment status to estimate survival time. The
latter approach is also available in the \href{https://cran.r-project.org/web/packages/RISCA/index.html}{RISCA} package \citep{Foucher2019} but is limited to Cox models.

\textbf{IPTW-IPCW Kaplan-Meier}

The \(\texttt{IPTW\_IPCW\_Kaplan\_meier}\) function implements the
IPTW-IPCW Kaplan Meier estimator from
Equation~\ref{eq-RMST_IPTW_IPCWKM}. It uses the utilitary functions from
the \(\texttt{utilitary.R}\) file
\(\texttt{estimate\_propensity\_score}\) and
\(\texttt{estimate\_survival\_function}\) to estimate the nuisance
parameters, and the function \(\texttt{adjusted.KM}\) which computes an
adjusted Kaplan Meier estimator using the appropriate weight.

\textbf{IPTW-BJ estimator}

The \(\texttt{IPTW\_BJ}\) implements the IPTW-BJ estimator in
Equation~\ref{eq-iptwbj}. It uses the utilitary functions, from the
\(\texttt{utilitary.R}\) file, \(\texttt{estimate\_propensity\_score}\),
\(\texttt{Q\_t\_hat}\) and \(\texttt{Q\_Y}\) to estimate the nuisance
parameters.

\textbf{AIPTW-AIPCW}

The \(\texttt{AIPTW\_AIPCW}\) function implements the AIPTW\_AIPCW
estimator in Equation~\ref{eq-AIPTW_AIPCW} using the utilitary function
from the \(\texttt{utilitary.R}\) file
\(\texttt{estimate\_propensity\_score}\), \(\texttt{Q\_t\_hat}\),
\(\texttt{Q\_Y}\), and \(\texttt{estimate\_survival\_function}\) to
estimate the nuisance parameters.

\subsection*{A.3. Appendix C: Data Generating Process and Figures}\label{appendix-c-figures}

\paragraph{Simulation coefficients (Section~\ref{sec-simulation-RCT}).}
In the RCT setting, we set \(\mu=(1,1,-1,1)\), \(\Sigma=\mathrm{Id}_4\), and use a Cox baseline rate \(\lambda_0=0.01\) with \(\beta_0=(0.5,0.5,-0.5,0.5)\) for \(\lambda^{(0)}(t\mid X)=\lambda_0\exp\{\beta_0^\top X\}\).
Censoring uses \(\lambda_{C,0}=0.03\): in Scenario~1, \(\lambda_C(t)=\lambda_{C,0}\); in Scenario~2, \(\lambda_C(t\mid X)=\lambda_{C,0}\exp\{\beta_C^\top X\}\) with \(\beta_C=(0.7,0.3,-0.25,-0.1)\).
Treatment is randomized with \(e(X)=0.5\).
\begin{figure}[!htbp]
\centering{
\hspace*{-0.05\textwidth}
\includegraphics[width=1.06\textwidth]{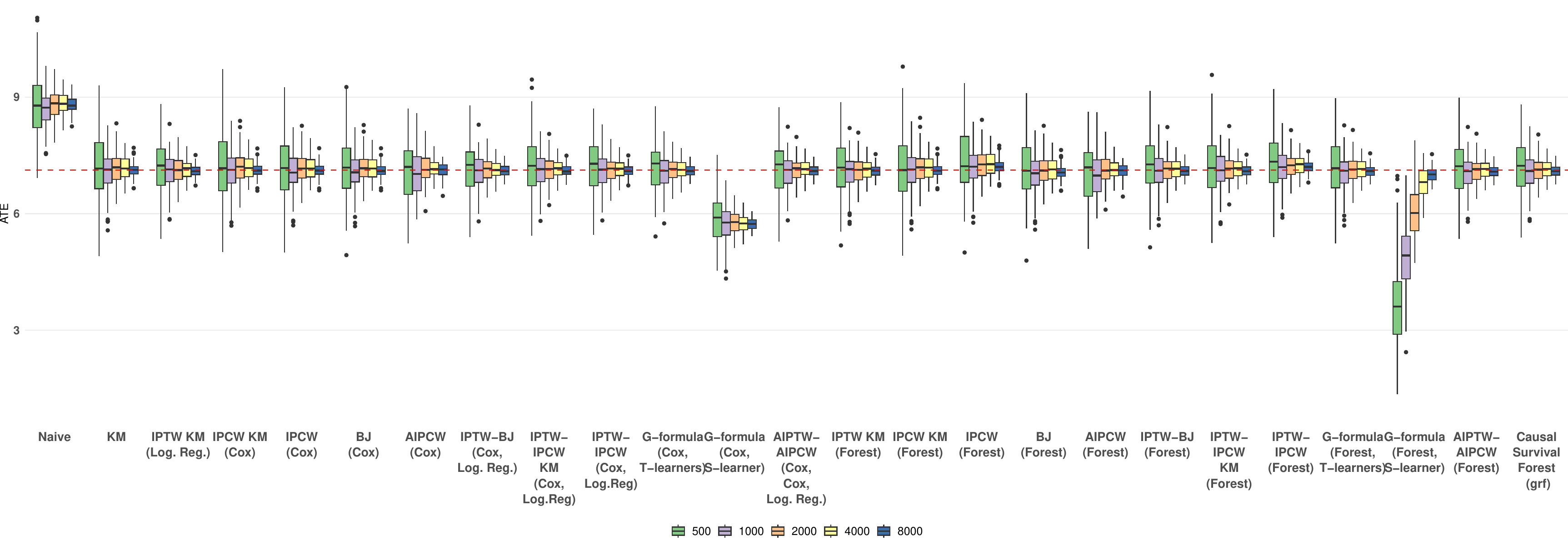}
}
\caption{\label{fig-rct1_full}Results of the ATE for the simulation of a RCT
with independent censoring.}
\end{figure}%

\begin{figure}[!htbp]
\centering{
\hspace*{-0.05\textwidth}
\includegraphics[width=1.06\textwidth]{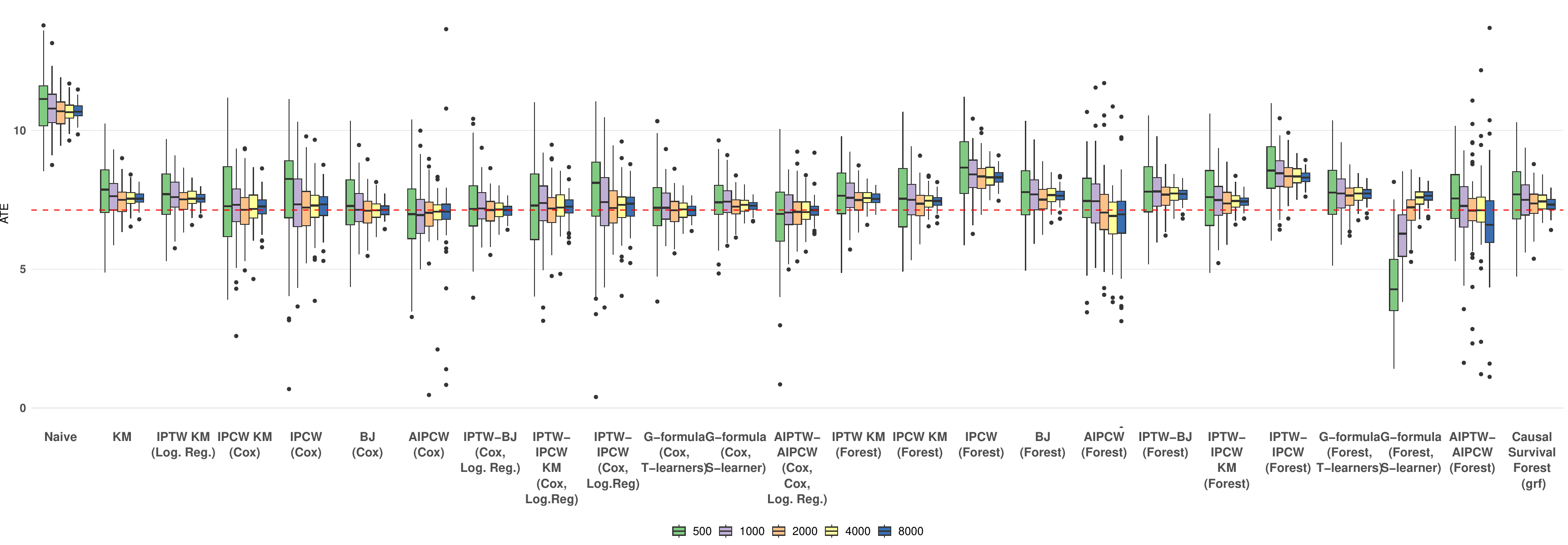}
}
\caption{\label{fig-rct2_full}Estimation results of the ATE for the
simulation of a RCT with conditionally independent censoring.}
\end{figure}%
\paragraph{Simulation coefficients (Section~\ref{sec-simulation-Obs}).}
In the observational setting, we use a logistic propensity model
\(\operatorname{logit}(e(X))=\beta_A^\top X\) with
\(\beta_A = (-1,\,-1,\,-2.5,\,-1)\).
Recall that \(\operatorname{logit}(p)=\log\!\big(p/(1-p)\big)\).
All other components of the data-generating process are identical to those in Section~\ref{sec-simulation-RCT}.

\begin{figure}[!htbp]
\centering{
\hspace*{-0.05\textwidth}
\includegraphics[width=1.06\textwidth]{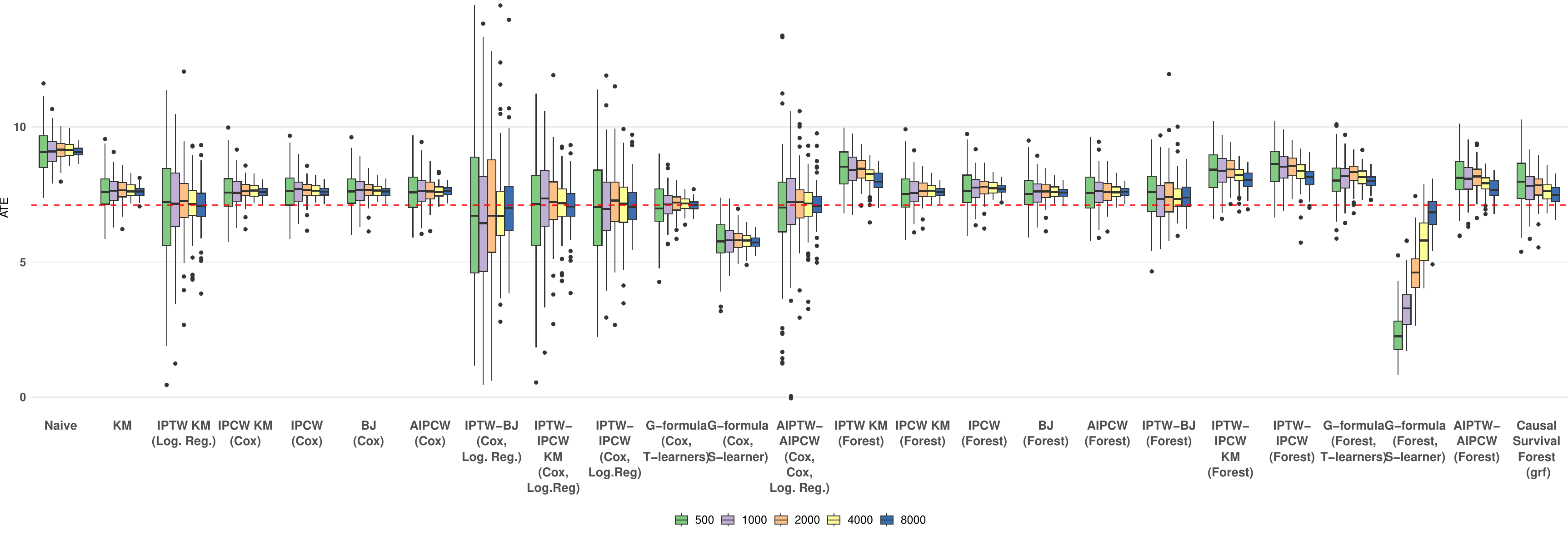}
}
\caption{\label{fig-obs1_full}Estimation results of the ATE for the
simulation of an observational study with independent censoring.}
\end{figure}%

\begin{figure}[!htbp]
\centering{
\hspace*{-0.05\textwidth}
\includegraphics[width=1.06\textwidth]{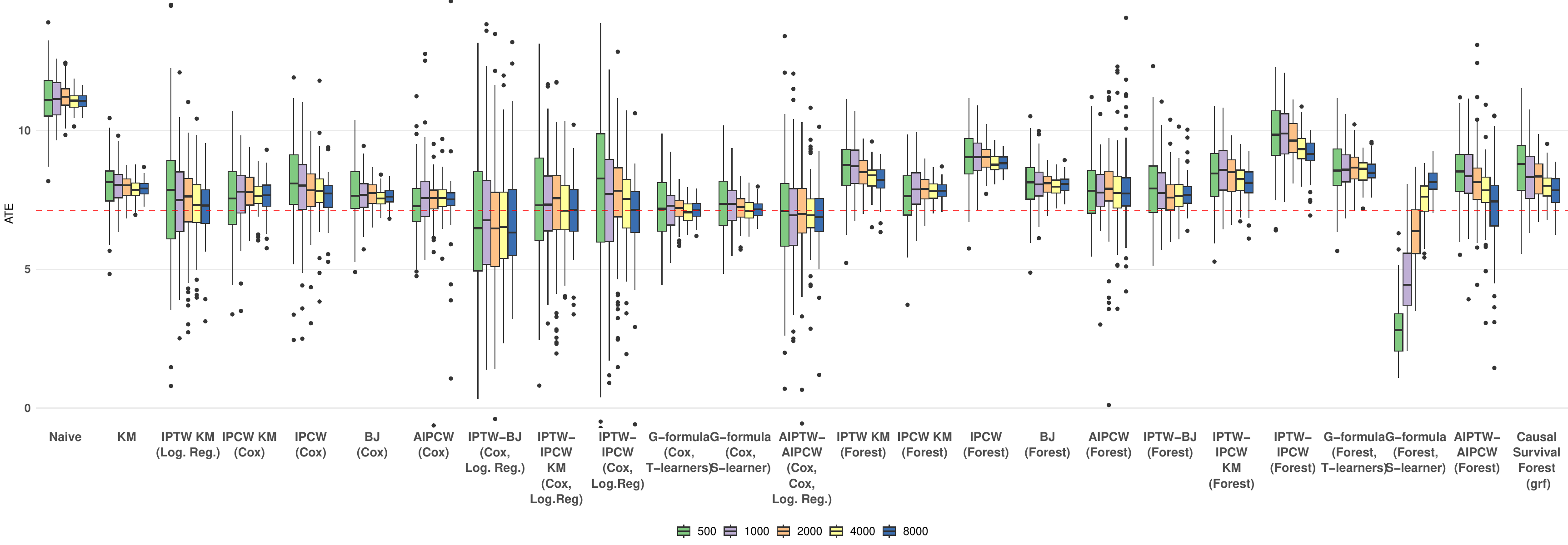}
}
\caption{\label{fig-obs2_full}Estimation results of the ATE for the
simulation of an observational study with conditionally independent censoring.}
\end{figure}%

\paragraph{Simulation coefficients (Section~\ref{sec-simulation-mis}).}
We set \(\mu=(0.1,0.5,0.7,0.4)\) and \(\Sigma=\mathrm{Id}_4\).
Let $Z(X)=(X_1^2,X_2^2,X_3^2,X_4^2,\; X_1X_2,X_1X_3,X_1X_4,X_2X_3,X_2X_4,X_3X_4).$
The control hazard uses \\ \(\beta_0=(0.2,\,0.3,\,0.1,\,0.1,\,0,\,1,\,0,\,0,\,0,\,-1)\).
Censoring uses \(\beta_C=(0.05,\,-0.01,\,0.05,\,-0.01,\,0,\,1,\,0,\,0,\,-1,\,0)\).
The propensity model uses \(\beta_A=(0.05,\,-0.1,\,0.5,\,-0.1,\,-1,\,0,\,-1,\,0,\,0,\,0)\).
Under well specification, nuisance models use inputs \((X,Z(X))\); under misspecification, they use only \(X\) and the first four components of \(Z(X)\), i.e., \((X_1^2,X_2^2,X_3^2,X_4^2)\).

\begin{figure}[H]

\centering{
\hspace*{-0.03\textwidth}
\includegraphics[width=0.95\textwidth]{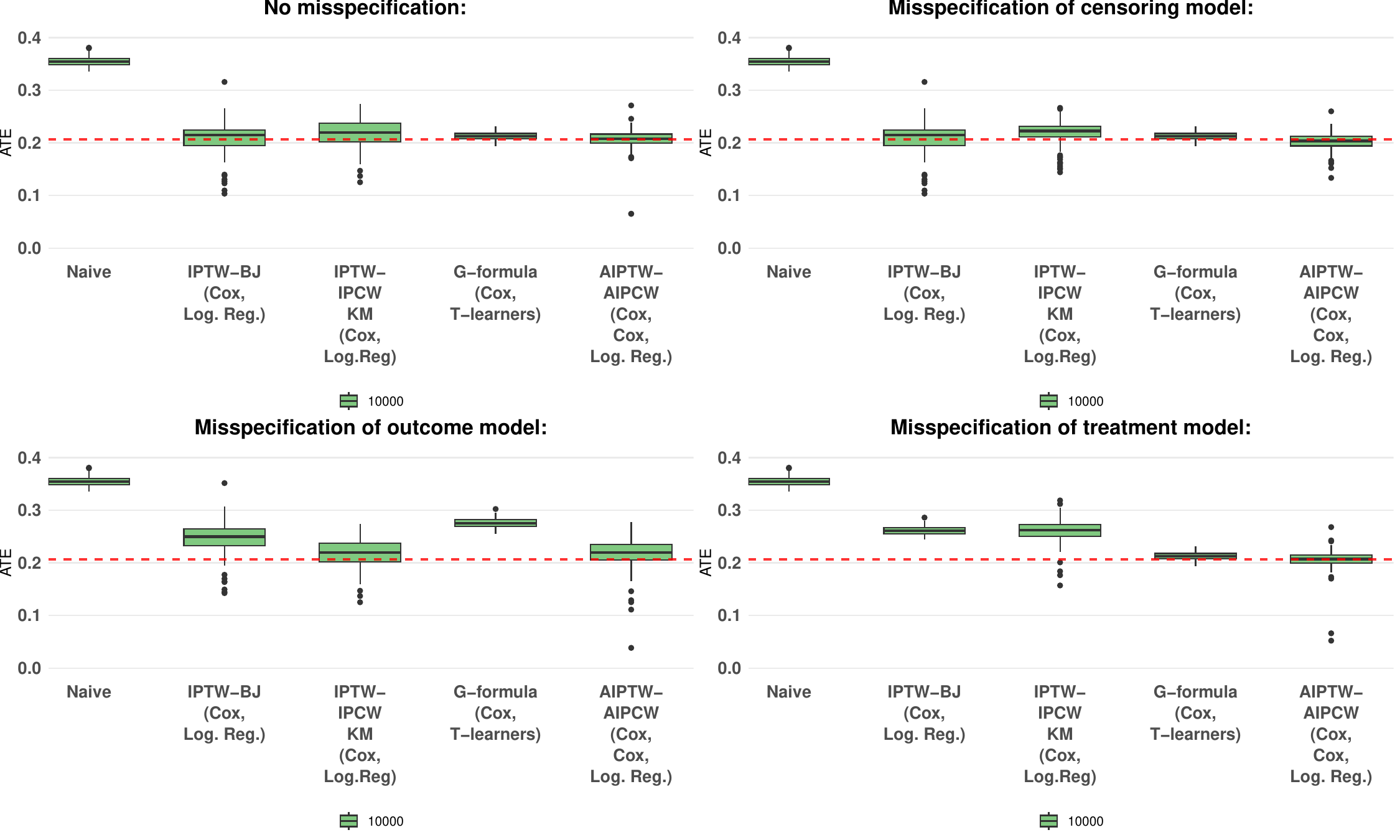}
}

\caption{\label{fig-mis2}Estimation results of the ATE for an
observational study with conditionally independent in case of a single
misspecification.}

\end{figure}%
\begin{figure}[h]
\centering{
\includegraphics[width=0.95\textwidth]{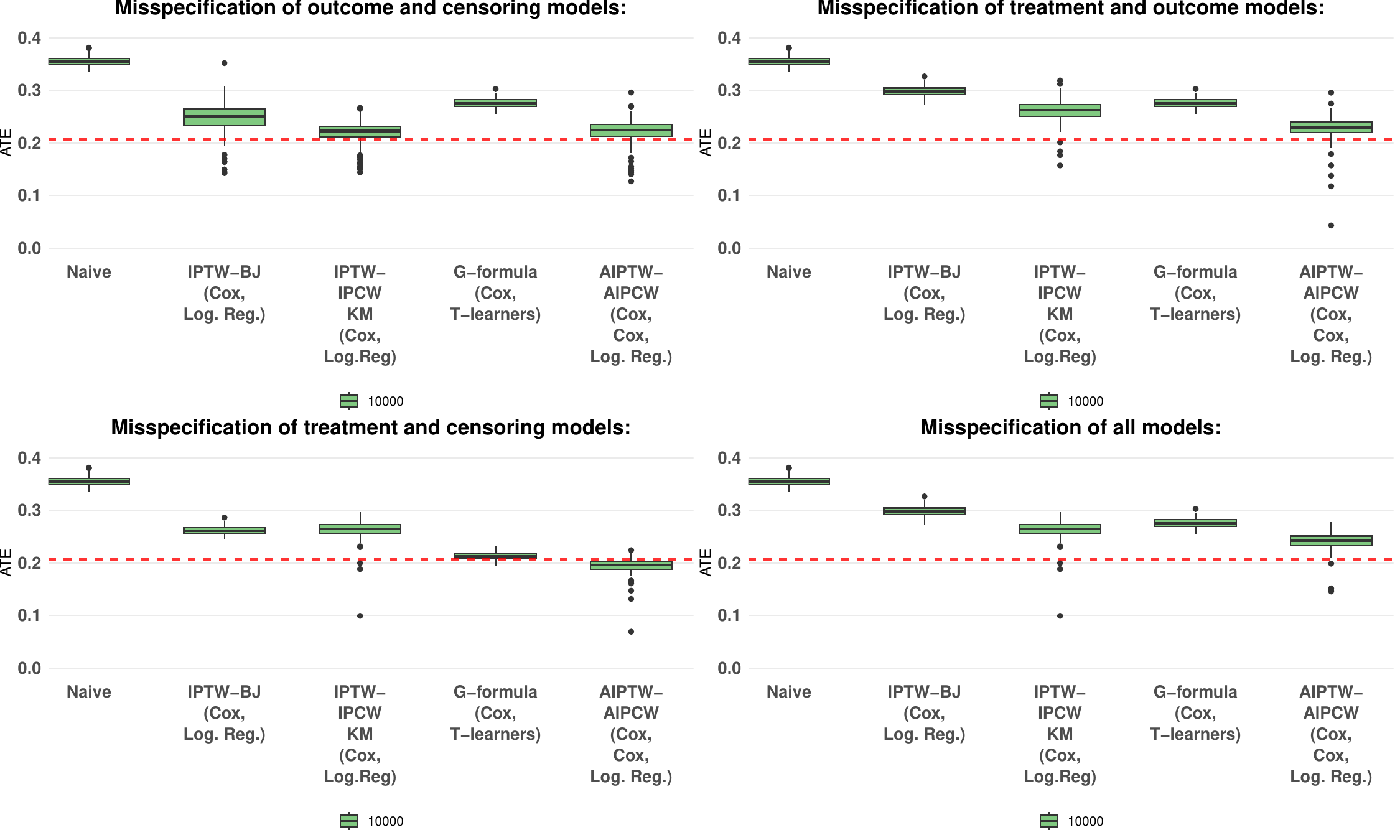}
}
\caption{\label{fig-mis3}{ATEs estimation in observational study with conditionally independent censoring with two or more misspecifications.}}
\end{figure}%

\end{document}